\documentclass[
 reprint,
 amsmath,amssymb,
 aps,
]{revtex4-2}

\usepackage{graphicx}
\usepackage{hyperref}
\usepackage[linesnumbered, ruled, vlined, titlenotnumbered]{algorithm2e}
\usepackage{amsmath, amsfonts, amssymb, amsthm, color}

\usepackage[capitalise]{cleveref}
\crefname{algocf}{alg.}{algs.}
\Crefname{algocf}{Algorithm}{Algorithms}

\newtheorem{theorem}{Theorem}

\newtheorem{corollary}{Corollary}
\newtheorem{lemma}{Lemma}

\newcommand{\ket}[1]{|#1\rangle}

\newcommand{\plog}{{p_{\log}}}

\DeclareMathOperator{\restart}{res}
\newcommand{\prestart}{{p_{\restart}}}
\DeclareMathOperator{\ovspace}{Q}
\DeclareMathOperator{\ovtime}{G}
\newcommand{\qubitoverhead}{{\omega_{\ovspace}}}
\newcommand{\gateoverhead}{{\omega_{\ovtime}}}
\DeclareMathOperator{\CliNR}{CliNR}
\DeclareMathOperator{\CZNR}{CZNR}

\newcommand{\N}{{\mathbb{N}}}

\begin{document}

\title{Low-cost noise reduction for Clifford circuits}

\author{Nicolas Delfosse}
\author{Edwin Tham}
\affiliation{
    IonQ Inc.
}

\date{\today}

\begin{abstract}
We propose a Clifford noise reduction (CliNR) scheme that provides a reduction of the logical error rate of Clifford circuit with lower overhead than error correction and without the exponential sampling overhead of error mitigation.
CliNR implements Clifford circuits by splitting them into sub-circuits that are performed using gate teleportation. A few random stabilizer measurements are used to detect errors in the resources states consumed by the gate teleportation.
This can be seen as a teleported version of the CPC scheme~\cite{roffe2018protecting, debroy2020extended, van2023single}, with offline fault-detection making it scalable.
We prove that CliNR achieves a vanishing logical error rate for families of $n$-qubit Clifford circuits with size $s$ such that $nsp^2$ goes to 0, where $p$ is the physical error rate, meaning that it reaches the regime $ns = o(1/p^2)$ whereas the direct implementation is limited to $s = o(1/p)$.
Moreover, CliNR uses only $3n+1$ qubits, $2s + o(s)$ gates and has zero rejection rate.
This small overhead makes it more practical than quantum error correction in the near term and our numerical simulations show that CliNR provides a reduction of the logical error rate in relevant noise regimes.
\end{abstract}

\maketitle

\section{Introduction}

Quantum error correction (QEC) and quantum error mitigation (QEM) are heavily-studied techniques for coping with imperfect quantum hardware.

Fault-tolerant quantum computing (FTQC) relies on quantum error correction~\cite{shor1995scheme} to achieve a vanishing logical error rate at constant physical error rate (below some threshold value)~\cite{aharonov1997fault, aliferis2006quantum}.
This approach is scalable as it comes with a polynomial asymptotic overhead in terms of qubit and gate count, but the error correction overhead is massive in practice.
For example, the factorization of RSA-2,048 may require surface codes using 1,457 physical qubits per logical qubit~\cite{gidney2021factor}.
In the asymptotic regime, constructions based on quantum LDPC codes and code concatenation achieve smaller overhead in terms of qubit count~\cite{gottesman2013fault, breuckmann2021quantum, panteleev2022asymptotically, leverrier2022quantum, pattison2023hierarchical, yoshida2024concatenate} or gate count~\cite{yamasaki2024time, zhou2024algorithmic}.
However, the cost remains daunting for explicit examples.
Recent simulations of quantum LDPC codes outperforming surface codes~\cite{tremblay2022constant, higgott2023constructions, bravyi2024high, scruby2024high} consume respectively 49, 32, 24 and 22 physical qubits per logical qubit, and extra qubits are needed to required for logical operations.
For comparison, the scheme proposed in this paper uses only 3 physical qubits per logical qubit.

Over the past few years, several building blocks of a fault-tolerant quantum computing architecture have been demonstrated experimentally~\cite{egan2021fault, ryan2021realization, postler2023demonstration, abobeih2022fault, krinner2022realizing, zhao2022realization, google2023suppressing, bluvstein2024logical, ryan2024high, mayer2024benchmarking}.
These experiments provide critical insights into the long-term goal of building a large-scale fault-tolerant quantum computer.
However, for current machines, it may be preferable to use more qubit-efficient techniques like quantum error mitigation.
Quantum error mitigation includes a variety of techniques reducing the impact of noise in a quantum circuit by collecting data from different variants of the circuit and post-processing this data~\cite{li2017efficient, temme2017error, cai2023quantum}.
It is already used in today's machines because it consumes only few or no extra qubits.
However, it is not a viable solution for large quantum circuits because the sampling cost, {\em i.e.}, the number of circuit executions required, generally grows exponentially with the size of the circuit.

In this work, we propose a scheme for noise reduction in Clifford circuits that fills the gap between the regime of FTQC and QEM. 
This Clifford noise reduction (CliNR) scheme is not fault-tolerant but it still reduces the logical error rate of Clifford circuits in practically relevant regimes, and it does so at the price of comparatively small overheads.

CliNR implements Clifford operations by gate teleportation~\cite{gottesman1999quantum}, which allows for offline preparation and checking of the associated stabilizer resource state.
A small number of random stabilizer measurements is used to detect errors in that resource state.
Gate teleportation is typically used to implement logical $T$ gates, together with $T$~state distillation~\cite{bravyi2005universal}.
However, fault-tolerance requires measuring a complete set of stabilizer generators of the $T$~state.
In CliNR, we measure only a few stabilizer generators to detect faults with a sufficiently large probability, ensuring the fault detection circuit remains small.

Our scheme can be seen as a teleported version of the coherent parity-check (CPC) scheme introduced in~\cite{roffe2018protecting}, which achieves a reduction of the logical error rate (by $50\%$ in \cite{debroy2020extended} and $34\%$ in \cite{gonzales2023quantum}) by measuring a few parity-checks.
Increasing the number of parity-checks in CPC leads to a more substantial reduction of the logical error rate, but results in an exponential sampling overhead~\cite{van2023single}.
The use of gate teleportation in CliNR avoids this issue because fault-detection is done offline -- an ancilla state is prepared and checked for faults, and then injected through gate teleportation only if no faults are observed.
If a fault is detected, only the ancilla state needs to be re-prepared, obviating the need to restart the entire computation from scratch.
This advantage holds with large Clifford circuits, as well as arbitrary quantum programs that can contain many Clifford sub-circuits, each of which can be realized through a CliNR procedure.

For large Clifford circuits, the sampling overhead can be kept in check through partitioning of the circuit into smaller sub-circuits with fewer gates, with each subcircuit being implemented using CliNR.
By adjusting the size of the sub-circuits, CliNR provides an extra flexibility to optimize the trade-off between the sampling overhead and the logical error rate.
Another advantage of our protocol is that the ancilla qubit used to measure the stabilizers is kept active for a shorter time than ancilla qubits in the CPC scheme, further reducing fault-detection noise stemming from idle errors.

We next consider the problem of implementing Clifford circuits with vanishing logical error rate and low overhead in the regime of vanishing physical error rate $p \rightarrow 0$.
The direct implementation provides a vanishing logical error rate for $n$-qubit Clifford circuits containing $s$ gates satisfying $sp \rightarrow 0$. 
\cref{theorem:snp_tradeoff} achieves a vanishing logical error rate for Clifford circuits with $nsp^2 \rightarrow 0$, meaning that the CliNR reaches circuits such that $ns = o(1/p^2)$, while the direct implementation requires $s = o(1/p)$.
\cref{cor:clifford_unitary_threshold} shows that CliNR implements arbitrary $n$-qubit Clifford unitaries (for which $s\lesssim n^2$) with vanishing logical error rate for all $n = O(p^{-2/3})$, whereas the direct implementation only reaches $O(p^{-1/2})$.
Notably, CliNR accomplishes this error suppression while consuming only $3n+1$ physical qubits and an average of $2s+o(s)$ single-qubit and two-qubit gates.

The remainder of this paper is organized as follows.
Background on quantum circuits is provided in \cref{sec:implementation}.
\cref{sec:clifford_noise_reduction} describes the CliNR circuit and presents our main results, \cref{theorem:CliNR_s0_r}, which provides a bounds on the performance of the CliNR circuits and \cref{theorem:snp_tradeoff}, which shows that one can implement Clifford circuits with $ns = o(1/p^2)$ with vanishing logical error rate.
Our numerical simulations are discussed in \cref{sec:numerics}.
\cref{app:cz_noise_reduction} considers a special case of CliNR, applicable to sequences of $CZ$ gates with reduced overheads, which could be relevant for chemistry applications~\cite{anand2023hamiltonians}.

\section{Implementation of a circuit}
\label{sec:implementation}

A {\em circuit with size} $s$ is a sequence of $s$ operations.
Each operation is either 
a preparation $\ket 0$ or $\ket +$,
a single-qubit unitary,
a controlled-Pauli ($CNOT = CX$, $CY$, or $CZ$), 
or the measurement of a qubit.
A measured qubit can be re-prepared and reused.

A {\em Clifford circuit} is a sequence of single-qubit Clifford unitaries and controlled-Paulis. 
These circuits take as an input a $n$-qubit state and output a $n$-qubit state.

We consider the standard circuit-level noise model with {\em physical error rate} $p$.
Preparations and single-qubit unitaries are followed by a Pauli error $X$, $Y$ or $Z$ with probability $p/3$ each.
Controlled-Paulis are followed by one of the 15 non-trivial Pauli errors acting on their support with probability $p/15$ each.
Measurement outcomes are flipped with probability $p$.
Errors corresponding to different operations are independent.
For now, we assume that all operations have the same error rate $p$ and that idle qubits are noiseless. 

Faults in a Clifford circuit propagate through the circuit, resulting in an error $E$ on the output state of the circuit that we refer to as the {\em output error}.
The output error can be computed by conjugating the circuit faults through the circuit~\cite{gottesman1998heisenberg}.

We say that a circuit $C'$ {\em implements} a $n$-qubit Clifford circuit $C$ if, in the absence of fault, for any input state, $C'$ produces the same output state as $C$.
The {\em logical error rate} for the implementation of $C$ with $C'$, denoted $\plog$, is defined to be the probability that faults propagate into a non-trivial output error.
In this definition, $C$ is the ideal circuit we wish to execute and $C'$ is the noisy circuit effectively implemented.
The circuit $C'$ is not restricted to Clifford unitaries and may contain measurements and ancilla qubits.
The extra qubits are traced-out at the end of the circuit to obtain a $n$-qubit output state as in $C$.

The number of qubits $n'$ and the size $s'$ of the implementation $C'$ are random variables because the circuit $C'$ may contains random or adaptive operations.
The {\em qubit overhead}, denoted $\qubitoverhead$, is defined as the ratio between the maximum number of qubits of the circuit $C'$ and the number of qubits of $C$.
The {\em gate overhead}, denoted $\gateoverhead$, is defined to be the ratio between the expected number of gates of $C'$ and the number of gates of $C$.
For the qubit overhead, we use the maximum number of qubits in $C'$ instead of the expectation because these qubits must physically exist in the machine, even if we do not use all of them at each execution of $C'$.

For example, the {\em direct implementation} of the circuit $C$ is the implementation $C' = C$.
Its qubit and gate overhead is $1$ and its logical error rate is upper bounded as 
$\plog \leq g_p(s)$
where $g_p(x) = 1 - (1-p)^x$ and $s$ is the size of $C$.

\section{Clifford noise reduction}
\label{sec:clifford_noise_reduction}

Throughout this section, $C$ is a $n$-qubit Clifford circuit and $r, t \geq 1$ are two integers.

\begin{figure}
    \centering
    \includegraphics[width=1\linewidth]{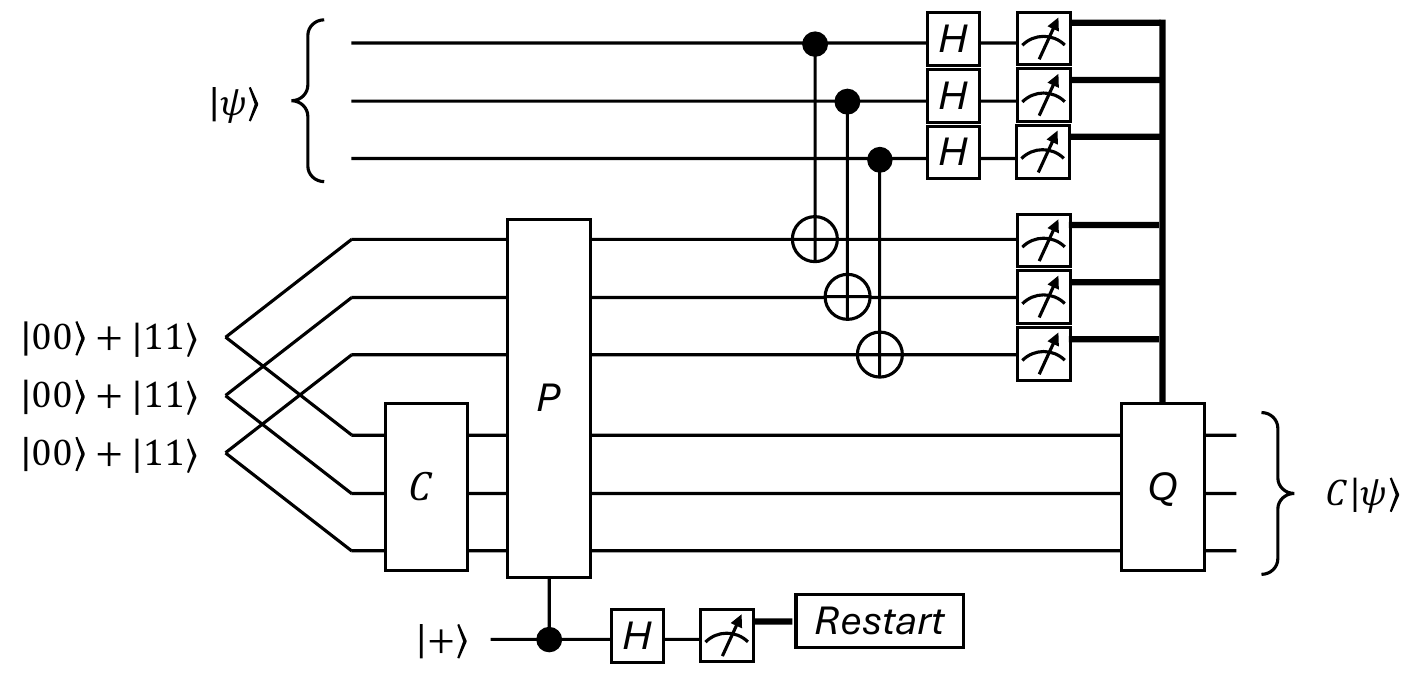}    
    \caption{
    The CliNR circuit (with $t=1$) implements a $n$-qubit Clifford circuit $C$ by gate teleportation using $2n+1$ ancilla qubits. 
    We insert $r$ stabilizer measurements before the CNOTs, to detect errors in the application of $C$ on the ancilla qubits. Here $r=1$ and we measure the stabilizer $P$. 
    If a stabilizer measurement returns a non-trivial outcome, the circuit restarts.
    Here $Q$ is a $n$-qubit Pauli operation that depends on the outcome of the measurement of the first $2n$ qubits.
    }
    \label{fig:clinr_circuit}
\end{figure}

Let us first define the {\em Clifford noise reduction} circuit, denoted $\CliNR_{t, r}(C)$, for $t=1$.
The circuit $\CliNR_{1, r}(C)$, represented in \cref{fig:clinr_circuit}, acts on $3n+1$ qubits.
We refer to the first three blocks of $n$ qubits as {\em first block}, {\em second block} and {\em third block} respectively and qubit $3n+1$ is an extra qubit.
The circuit $\CliNR_{1, r}(C)$ is made with the following operations.
\begin{enumerate}
    \item Prepare $n$ Bell states on qubits $n+i$ and $2n+i$ for $i=1,\dots, n$.
    \item Apply $C$ to the third block.
    \item Measure $r$ random stabilizers of the state supported on blocks two and three.
    \item If a stabilizer measurement returns a non-trivial outcome, restart the circuit.
    \item Apply $CNOT_{i, n+i}$, followed by $H_i$, for $i=1,\dots, n$.
    \item Measure the first and second block and apply the operation $Q$ (defined below) to the third block.
\end{enumerate}
The stabilizers are measured sequentially using qubit $3n+1$ and a sequence of up to $2n$ controlled-Pauli gates.
Following~\cite{van2023single}, we select the $r$ independent stabilizers uniformly at random so that each stabilizer measurement detect $50\%$ of the faults.
A new set of $r$ random stabilizer generators is selected for each run of the circuit.

If a stabilizer measurement returns a non-trivial outcome, we know that a fault has occurred and we restart the circuit.
Crucially, these faults do not affect the input state $\ket{\psi}$ because they are detected before the CNOTs connecting the first block to the last $2n+1$ qubits.

The final classically-controlled operation is 
\begin{align}
    Q = \prod_{i=1}^n \left( C^\dagger X_{2n+i} C \right)^{o_{n+i}} \left( C^\dagger Z_{2n+i} C \right)^{o_i}
\end{align}
where $o_i$ is the outcome of the measurement of the $i$ th qubit.

The main limitation of the CliNR implementation
with $t=1$ is that the restart probability increases rapidly as circuit size $s$ approaches $1/p$, resulting in a large gate overhead.
To maintain a low cost for larger circuits, we decompose the input circuit $C$ into $t$ smaller sub-circuits and we apply $\CliNR(1, r)$ to each of these sub-circuits.

For $t \geq 1$, the circuit $\CliNR_{t, r}(C)$ is defined by writing $C$ as the concatenation of $t$ sub-circuits $C_1, \dots, C_t$ and by applying $\CliNR_{1, r}$ to each sub-circuit.
In other words, it is the concatenation of the circuits $\CliNR_{1, r}(C_1), \dots, \CliNR_{1, r}(C_t)$. Formally, we need to swap the first and third blocks to concatenate two such circuits. These swaps can be done fully in software by relabelling the qubits.

In what follows, we denote $s_0 := \lceil s/t \rceil$.
If $s$ is a multiple of $t$, each sub-circuit has size $s_0$.
Otherwise, the first $s \mod t$ sub-circuits have size $s_0$ and the remaining ones have size $s_0 - 1$.

\begin{theorem}
\label{theorem:CliNR_s0_r}
    The circuit $\CliNR_{t, r}(C)$ implements the Clifford circuit $C$ with logical error rate
    \begin{align}
    \label{eq:theorem:CliNR_s0_r:plog_bound}
        \plog \leq & t \cdot \frac{g_p(3n+s_0) 2^{-r} + 2 g_p(2n + 3) + g_p(5n)}{(1-p)^{m_0}}
    \end{align}
    where $m_0 = 3n + s_0 + (2n + 3)r$.
    Moreover, the overhead satisfies $\qubitoverhead = 3 + \frac{1}{n}$ and
    $\gateoverhead \leq \frac{10n}{s_0} + \frac{2m_0}{s_0(1-p)^{m_0}}$.
\end{theorem}

This result is proven in \cref{appendix:proof_theorem}.
This bound on the logical error rate only holds for Clifford circuits because then $Q$ is comprised of just $n$ single-qubit Pauli gates. 
When $C$ is non-Clifford, $Q$ is more complex and generally consists of longer sequence of noisy gates.

\section{The $snp^2$ threshold}
\label{sec:threshold}

Consider the problem of implementing $n$-qubit Clifford circuits with vanishing logical error rate and low overhead, in the regime where the physical error rate $p$ goes to 0.

We are given a family $(C_j)_{j \in \N}$ of Clifford circuits, acting on $n_j$ qubits with size $s_j$, implemented with qubits with physical error rate $p_j$.
To keep the notation simple, we omit the index $j$ and we write $sp \rightarrow 0$ and $snp^2 \rightarrow 0$.

The expected number of faults in a circuit with size $s$ is $sp$. Therefore, if $sp \rightarrow 0$, the direct implementation achieves a vanishing logical error rate. 
The following result shows that CliNR goes beyond this threshold. 

\begin{theorem} 
[$snp^2$ threshold]
\label{theorem:snp_tradeoff}
    Consider a family of $n$-qubit Clifford circuits with size $s$.
    If $snp^2 \rightarrow 0$, then we can implement these circuits with     vanishing logical error rate.
    Moreover, we have $\limsup \qubitoverhead \leq 3$ and $\limsup \gateoverhead \leq 2$.
\end{theorem}

This result reaches circuits with size up to $s = o(1/p^2)$, whereas the direct implementation achieves a vanishing logical error rate only if $s = o(1/p)$.
It is based on the CliNR implementation with parameters $t, r$ given in \cref{eq:proof:snp_tradeoff:def_r_t}.

\begin{proof}
Let $\varepsilon := snp^2$. 
By assumption, we have $\varepsilon \rightarrow 0$.
If  
$
s < \frac{\varepsilon^{1/4}}{p}
$, 
then $s = o(1/p)$ and the direct implementation is enough to achieve a vanishing logical error rate.

In the remainder of this proof, assume that 
$s \geq \frac{\varepsilon^{1/4}}{p}$.
Then, we have
\begin{align}
\label{eq:proof:snp_tradeoff:n_upper_bound}
n = \frac{\varepsilon}{s p^2} \leq \frac{\varepsilon^{3/4}}{p} = o\left(\frac{1}{p} \right)
\end{align}
and 
\begin{align}
\label{eq:proof:snp_tradeoff:n_over_s_vanishes}
\frac{n}{s} \leq \sqrt{\varepsilon} \rightarrow 0 \cdot
\end{align}

Consider $C' = \CliNR_{t, r}(C)$ with 
\begin{align}
\label{eq:proof:snp_tradeoff:def_r_t}
    r := \left \lfloor \log( \frac{s}{n} ) \right \rfloor \text{ and }
    t := \left \lfloor \sqrt{\frac{s}{n}} \right \rfloor \cdot
\end{align}
We want to apply \cref{lemma:asymptotic_bounds}, proven in \cref{appendix:proof_lemma}.
Let us first prove that this result applies.
Because $s_0$ is proportional with $\frac{s}{t}$, we have
\begin{align}
\label{eq:proof:snp_tradeoff:rn_over_s0_vanishes}
    \frac{rn}{s_0} = O\left( \frac{rnt}{s} \right)
    = O\left( \log\left(\frac{s}{n}\right) \sqrt{\frac{n}{s}} \right)
\end{align}
which goes to $0$ based on \cref{eq:proof:snp_tradeoff:n_over_s_vanishes}.
Moreover,
\begin{align}
\label{eq:proof:snp_tradeoff:s0p_vanishes}
    s_0 p = O\left( \frac{sp}{t} \right)
    = O\left( \sqrt{sn} \cdot p \right)
\end{align}
which goes to $0$ because $snp^2$ vanishes.
Because $\frac{rn}{s_0}$ and $s_0 p \rightarrow 0$ we can use \cref{lemma:asymptotic_bounds}.

To show that the logical error rate of the implementation $C'$ vanishes, it is sufficient to prove that the bound in \cref{eq:lemma:asymptotic_bounds:plog} tends to $0$.
Based on \cref{eq:proof:snp_tradeoff:s0p_vanishes}, the denominator goes to $1$, hence it is enough to show that 
$t s_0 p 2^{-r}$ and $tnp$ go to 0.
Using $t s_0 = O(s)$ and $2^{-r} = O(\frac{n}{s})$, we obtain
$t s_0 p 2^{-r} = O\left( np \right)$
which goes to 0 by \cref{eq:proof:snp_tradeoff:n_upper_bound}.
For $tnp$, we get
$
tnp 
= O\left( \sqrt{sn}p \right)
$
which goes to 0 because $snp^2$ does.
\end{proof}

\begin{corollary} 
[size-$n^\alpha$ threshold]
\label{cor:clifford_alpha_threshold}
Consider a family of $n$-qubit Clifford circuits with size $s = O(n^{\alpha})$ for some $\alpha > 0$.
If $n p^{\frac{2}{1 + \alpha}} \rightarrow 0$, 
we can implement these circuits with vanishing logical error rate.
Moreover, we have $\limsup \qubitoverhead \leq 3$ and $\limsup \gateoverhead \leq 2$.
\end{corollary}

This result is an immediate application of \cref{theorem:snp_tradeoff}. It applies to circuits acting on up to $n = O(p^{- \frac{2}{1 + \alpha}})$ qubits, whereas the direct implementation only reaches $n = O(p^{- \frac{1}{\alpha}})$.


\begin{corollary} 
[Clifford unitary threshold]
\label{cor:clifford_unitary_threshold}
If $p \rightarrow 0$ and $n = O(p^{- \frac{2}{3}})$, we can implement arbitrary $n$-qubit Clifford unitaries with vanishing logical error rate. 
Moreover, we have $\limsup \qubitoverhead \leq 3$ and $\limsup \gateoverhead \leq 2$.
\end{corollary}

This is an application of \cref{theorem:snp_tradeoff}, together with the fact that any $n$-qubit Clifford unitary can be implemented with a Clifford circuit $C$ with size 
$
s = O\left(n^2/\log(n) \right) = o(n^2)
$~\cite{aaronson2004improved}.
Again, this goes beyond the direct implementation which only reaches $n = O(p^{-\frac{1}{2}})$.


\section{Numerical results}
\label{sec:numerics}

We numerically estimate the logical error rate of CliNR for the implementation of uniform random Clifford unitaries, sampled using Algorithm~2 of~\cite{bravyimaslov2021}, synthesized into $s \approx n^2$ gates from the set {$H_i$, $S_i$, $CNOT_{i, j}$}.

We consider the circuit-level noise model with noise rate $p_2$ for two-qubit operations and $p_1 = p_2/10$ for single-qubit operations because two-qubit gates are typically the noisiest operations. 
Moreover, we include noise on idle qubits with rate $p_1$.

In \cref{fig:clinr_circuit}, we set $p_2 = 10^{-3}$ and $10^{-4}$.
Ten random Clifford circuits are generated for each $n$, and each circuit is simulated over $10^5$ shots with the direct implementation and $10^5$ shots with the CliNR implementation.
We observe that CliNR outperforms the direct implementation in relevant noise regimes.
We obtain a $2 \times$ reduction of the logical error rate for the implementation of 25-qubit Clifford circuits when $p_2 = 10^{-3}$.
For $p_2 = 10^{-4}$, we get a $4 \times$ reduction of the logical error rate for Clifford circuits on 60 qubits.

\begin{figure}
    \centering
    \includegraphics[width=1\linewidth]{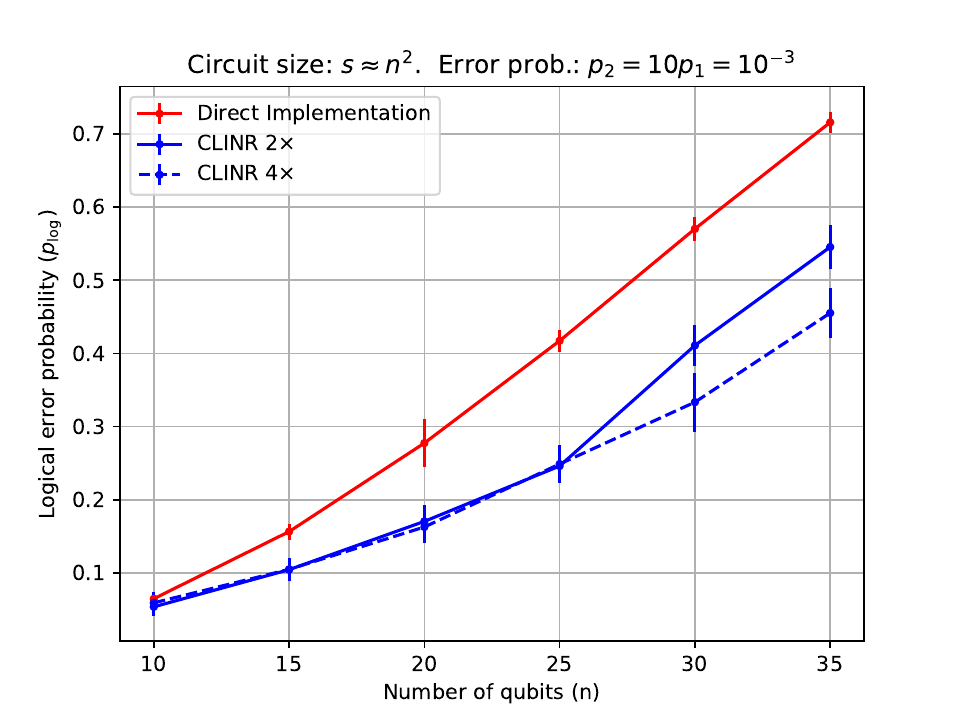}
    
    \includegraphics[width=1\linewidth]{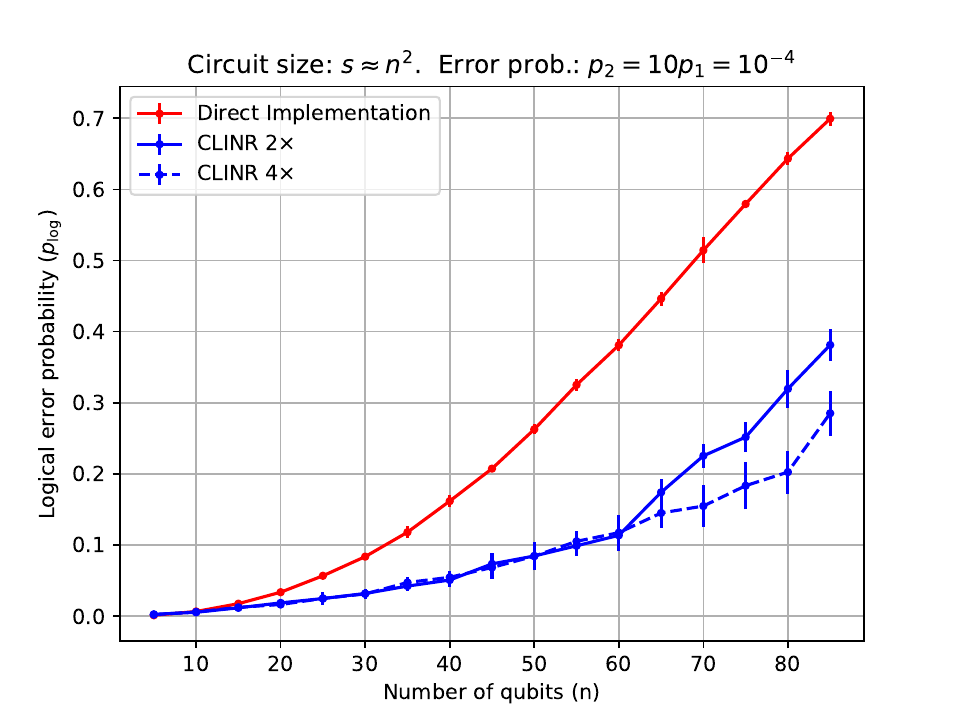}
    
    \caption{Averge logical error rate of the CliNR implementation of random Clifford circuits in two different noise regimes.
    We chose $r = \left \lfloor \log( \frac{s}{n} ) \right \rfloor$, as in \cref{eq:proof:snp_tradeoff:def_r_t}, and $t$ is the smallest integer such that the $\gateoverhead \leq 2$ for CliNR $2 \times$ and $\gateoverhead \leq 4$ for CliNR $4 \times$.
    The value of $t$ varies with $n$.}
    \label{fig:plots_realistic_noise}
\end{figure}

Instead of using stabilizers selected uniformly at random to detect faults, like the proof suggests, we restrict ourselves to stabilizers obtained by propagating weight-2 stabilizers of the Bell state $\frac{1}{\sqrt{2}}(\ket{00}+\ket{11})^{\otimes n}$ through $I \otimes U$.
The resulting stabilizers have weight at most $n+1$ instead of $2n$, making their measurements less noisy.
Clifford circuits with unevenly distributed noisy gates may require a modified choice of measured stabilizers better adapted to its shape and noise profile.

It is worth noting, we do not select a new set of stabilizers for every run, but do so in batches of $10^3$ shots instead, for expedience while performing numerical simulations.
This is justified for random Clifford circuits because they rapidly scramble faults as they propagate, but may not work well for shallower structured circuits.
We leave a more careful analysis of such circuits for a later work.

\section{Conclusion}
\label{sec:conclusion}

We proposed a method for error reduction in Clifford circuits offering a low-cost alternative to error correction in near term devices.
Clifford circuits are split into small sub-circuits that are implemented by gate teleportation. 
Faults are detected by measuring stabilizers of the resource states consumed by the gate teleportation circuit. 
Instead of stabilizer measurements, one could verify the resource states by measuring gauge operators of the spacetime code of the circuit~\cite{bacon2015sparse, delfosse2023spacetime}.
This paper focuses on fully connected qubits but the scheme can be adapted to other qubit topologies at a moderate price in term of overhead as demonstrated in~\cite{van2023single}.

This scheme is limited to Clifford circuits that are not computationally universal.
We envision applications wherein a suitable compiler collates Clifford gates within a given ``real-world'' circuit in order to yield large Clifford sub-circuits separated by thin layers of single-qubit gates, which is sufficient for universality.
Alternatively, CliNR can be applied to error-corrected qubits with universality achieved through magic states~\cite{bravyi2005universal}.

\section{Acknowledgement}

The authors thank Aharon Brodutch, Laird Egan, John Gamble, Andrii Maksymov, Jeremy Sage, Pat Tang and Brendan Wyker for insightful discussions.

\appendix

\section{Gate teleportation}

The gate teleportation circuit~\cite{gottesman1999quantum} is a key ingredient of the CliNR scheme.
\cref{fig:clinr_circuit}(a) shows the standard gate teleportation circuit for one qubit.
It performs simultaneously the teleportation of the state $\ket \psi$ from qubit 1 to qubit 3, and the application of a unitary gate $U$ to $\ket \psi$.

The one-bit teleportation circuit of \cref{fig:teleportation}(b), introduced in~\cite{zhou2000methodology}, performs the teleportation of a quantum state $\ket \psi$ without using any ancilla qubit.
If $I \otimes U$ commutes with the $CNOT$ gate of the one-bit teleportation circuit, we can play the same trick as in the original teleportation circuit to merge the one-bit teleportation and the application of $U$.
This results in a circuit that uses the ancilla state $U \ket +$ to apply the gate $U$ to $\ket \psi$.
In \cref{app:cz_noise_reduction}, we use this idea to build a state injection circuit for sequences of $CZ$ gates.

\begin{figure}[h]
    \centering

    \includegraphics[width=1\linewidth]{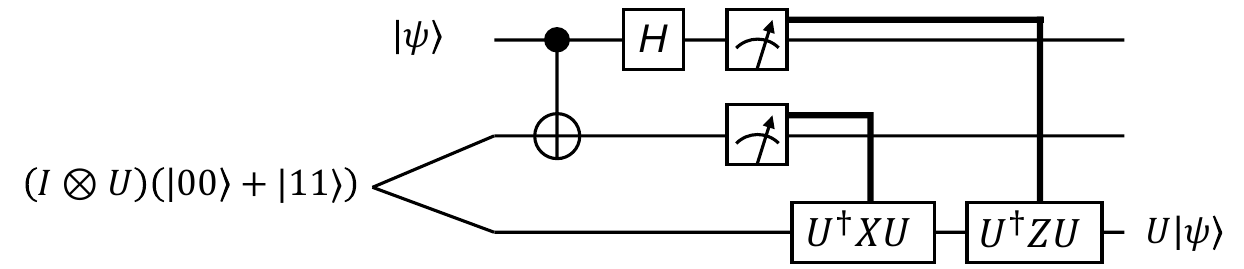}    

    (a)

    \vspace{.5cm}
    \includegraphics[width=0.5\linewidth]{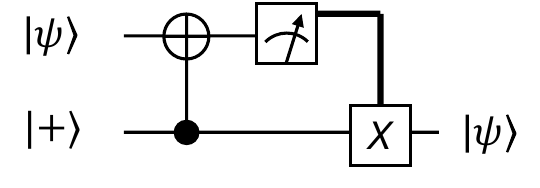}    
    
    (b)
    
    \caption{
    (a) Gate teleportation circuit implementing $U$~\cite{gottesman1999quantum}.
    When $U = I$, we recover the original teleportation circuit~\cite{bennett1993teleporting}.
    (b) One-bit teleportation circuit~\cite{zhou2000methodology}.
    }
    \label{fig:teleportation}
\end{figure}

\section{Bound on the performance of CliNR}
\label{appendix:proof_theorem}

The goal of this section is to prove \cref{theorem:CliNR_s0_r}.
We start with the following lemma, which bounds the performance of CliNR in the special case $t=1$.

\begin{lemma}
\label{lemma:CliNR}
    The circuit $\CliNR_{1, r}(C)$ implements the circuit $C$ with logical error rate
    \begin{align}
        \plog \leq & \frac{g_p(3n+s) 2^{-r} + 2g_p(2n + 3) + g_p(5n)}{(1-p)^m}
    \end{align}
    where $m = 3n + s + (2n + 3)r$.
    Moreover, the overhead satisfies $\qubitoverhead = 3 + \frac{1}{n}$ and $\gateoverhead \leq  \frac{5n}{s} + \frac{m}{s(1-p)^{m}}$.
\end{lemma}

\begin{proof}
Denote $C' = \CliNR_{1, r}(C)$.
The fact that $C'$ implements $C$ is clear because, in the absence of fault, the stabilizer measurements have no effect, and removing these measurements, we obtain the standard gate teleportation circuit.

Let us bound the overhead of the implementation.
The qubit overhead is clear.
To bound the gate overhead, we need an upper bound on the probability $\prestart$ of $C'$ to restart during the stabilizer measurements.
Naively, it is upper bounded by the probability that at least one fault occurs before the end of the stabilizer measurement loop. 
This leads to
\begin{align}
\label{eq:proof:CliNR:prestart_bound}
\prestart \leq 1 - (1-p)^{m}    
\end{align}
where $m = 3n + s + (2n + 3)r$ is the number of circuit operations before the last stabilizer measurement (included).
Each restart adds at most $m$ operations to the circuit.
Moreover, if no restart is triggered, $5n$ operations are added to finish the gate teleportation.
Therefore, the expected number of operations executed by $C'$ is upper bounded by 
\begin{align}
    & 5n + \sum_{i \geq 1} i m \cdot \prestart^{i-1} (1 - \prestart) \\
    & \leq 5n + \frac{m}{1 - \prestart} \\
    & \leq  5n + \frac{m}{(1-p)^{m}} \cdot
\end{align}
The gate overhead is given by dividing this bound by the size $s$ of $C$.

To bound the logical error rate of the implementation of $C$ by $C'$, assume that a combination of faults occurs in the circuit, resulting in a output error $E$.
Below, we say a combination of faults is non-trivial if they propagate into a non-trivial output error.
If the output error $E$ is non-trivial, then, one of the following events must have occurred.

(1) Non-trivial faults occur before the beginning of the stabilizer measurements and they are not detected by the stabilizer measurements.

(2) Non-trivial faults occur during the $i$ th stabilizer measurement and they are not detected by the subsequent stabilizer measurements

(3) Non-trivial faults occur after the last stabilizer measurement.

Denote by $P_1, P_2$ and $P_3$ the respective probabilities of these three events.
Then, the logical error rate of the CliNR implementation is upper bounded by 
\begin{align}
\label{eq:proof:CliNR:plog_bound}
\plog \leq \frac{P_1 + P_2 + P_3}{1-\prestart} \cdot
\end{align}
Next, we bound the three terms in the numerator.

Case (1). The circuit $C'$ contains $3n + s$ noisy operations before the first stabilizer measurement.
The probability that at least one fault occurs in these operations is $1 - (1-p)^{3n + s}$.
If these faults are non-trivial, the outcome $\sigma \in \{0,1\}^r$ for the $r$ stabilizer measurements is a uniform random bit-string.
This distribution remains uniform when including the effect of fault occurring during the measurement.
Therefore, we have 
\begin{align}
\label{eq:proof:CliNR:P1_bound}
    P_1 \leq (1 - (1-p)^{3n + s}) 2^{-r}
\end{align}
because the faults are undetected only if $\sigma = 0$.

Case (2). Each stabilizer measurement contains at most $2n+3$ operations.
By the same argument as in (1), a non-trivial combination of faults in the $i$ th measurement induces a uniform outcome on the last $r-i$ stabilizer measurements.
Therefore, it is undetected with probability at most $2^{-r+i}$, which yields
\begin{align}
\label{eq:proof:CliNR:P2_bound}
    P_2 
    & \leq \sum_{1=1}^{r} (1 - (1-p)^{2n + 3}) 2^{-r+i} \\
    & \leq 2 (1 - (1-p)^{2n + 3})
\end{align}

Case (3). Counting the circuit operations after the last measurement, we get
\begin{align}
\label{eq:proof:CliNR:P3_bound}
    P_3 \leq (1 - (1-p)^{5n}) \cdot
\end{align}
Injecting \cref{eq:proof:CliNR:prestart_bound,,eq:proof:CliNR:P1_bound,,eq:proof:CliNR:P2_bound,,eq:proof:CliNR:P3_bound} in \cref{eq:proof:CliNR:plog_bound} provides the bound on $\plog$.
\end{proof}

\begin{proof}[Proof of \cref{theorem:CliNR_s0_r}]
The fact that $\CliNR_{t, r}(C)$ implements $C$ is a consequence of \cref{lemma:CliNR} which shows that $\CliNR_{1, r}(C_i)$ implements the sub-circuit $C_i$.

The logical error rate is at most $t$ times the logical error rate per sub-circuit, for which we use the bound from \cref{lemma:CliNR}.

The qubit overhead is the same as in $\CliNR_{1, r}$ because the $2n+1$ ancilla qubits are re-used for the implementation of each sub-circuit. 
The expected number of operations in the implementation of one of the sub-circuits is at most $s_0 \gateoverhead'$ where $\gateoverhead'$ is the maximum gate overhead for a sub-circuit.
By linearity of the expectation, the expected number of operations for the implementation of the entire circuit $C$ is upper bounded by 
$t s_0 \gateoverhead'$.
Using $t s_0 \leq t \lceil s/t \rceil \leq s + t \leq 2s$ (if $t \geq s$, the last $t-s$ sub-circuits $C_i$ have size 0), we get
\begin{align}
\gateoverhead 
    \leq \frac{t s_0 \gateoverhead'}{s}
    \leq 2 \gateoverhead'
\end{align}
Plugging in the bound on the gate overhead obtained in \cref{lemma:CliNR} leads to the bound announced in \cref{theorem:CliNR_s0_r}.
\end{proof}

\section{Proof of \cref{lemma:asymptotic_bounds}}
\label{appendix:proof_lemma}

The goal of this section is to prove the following lemma.

\begin{lemma}
\label{lemma:asymptotic_bounds}
If $nr = o(s_0)$ and $s_0p \rightarrow 0$, then
\begin{align}
\label{eq:lemma:asymptotic_bounds:plog}
    \plog = 
    O\left(
        \frac{t s_0 2^{-r} p + 9tnp}{1 - s_0 p} \cdot
    \right)
\end{align}
Moreover, the overhead satisfies $\limsup \qubitoverhead \leq 3$ and $\limsup \gateoverhead \leq 2$.
\end{lemma}

\begin{proof}
We compute the asymptotic equivalent of the upper bound in \cref{eq:theorem:CliNR_s0_r:plog_bound} when $s_0 p$ vanishes.
The expansion $g_p(\ell) = \ell p + o\left( (\ell p)^2 \right)$ leads to a numerator dominated by 
\begin{align}
t \left( (A s_0 + A' n)2^{-r} + (Bn + B') \right) p
\end{align}
with $A = 1, A' = 3, B = 9$ and $B' = 6$.
We can discard the term $A' n$, dominated by $A s_0$ (because we assume $nr = o(s_0)$) and the term $B'$, dominated by $Bn$.
For the denominator, we use $(1-p)^m = 1 - mp + o((mp)^2)$.
This proves \cref{eq:lemma:asymptotic_bounds:plog}.

The bound on $\qubitoverhead$ is clear from \cref{theorem:CliNR_s0_r}.
In the upper bound on $\gateoverhead$, the term $\frac{10n}{s_0}$ vanishes and $\frac{2m_0}{s_0 (1-p)^{m_0}}$ goes to $2$ because 
$\frac{2m_0}{s_0} \rightarrow 2$ and 
$\frac{1}{(1-p)^{m_0}} = e^{-m_0p + o((m_0p)^2)} \rightarrow 1$. 
\end{proof}

\section{CliNR and Circuit Shape}
\label{app:circuit_shape}

The efficacy of CliNR at reducing noise varies significantly with the circuit's shape parameter $\alpha$, where the circuit size is given by $s=n^{\alpha}$.
At a given width~$n$, smaller $\alpha$ implies a circuit with fewer gates in the intended Clifford block $C$, even as the overheads associated with CliNR remain independent of $\alpha$.
Thus, for any given $n$, circuits with smaller $\alpha$ are more likely to favor direct implementation, whereas CliNR becomes increasingly effective as $\alpha \to 2$. Nevertheless at sufficiently large qubit counts ($n>p^{-1/{\alpha+1}}$), direct implementation becomes very noisy, whereas from \cref{cor:clifford_alpha_threshold} we expect a CliNR implementation to remain effective.

In Fig.~\ref{fig:Grid}, we plot differences in the logical error rate between CliNR and direct implementation of Clifford circuits of various sizes.
Each cell in the color-coded grids represent random Clifford circuits of distinct $n$ and $\alpha$.
Random Clifford circuits at a given $n$ and $\alpha$ are constructed by randomly sampling $n^{\alpha}$ gates from the set \{$H$, $S$, $CX$\}.

Logical error rate differences, $p_{\text{log}}^{\text{direct}}-p_{\text{log}}^{\text{CliNR}}$ are depicted as a color map.
At small $n$ and $\alpha$, we observe direct implementation yielding error rates competitive with CliNR, without the overheads of the latter.
As $\alpha$ or $n$ increase, CliNR becomes strongly favored with substantial error rate improvements over direct implementation.

\begin{figure}
    \centering
    \includegraphics[width=7cm]{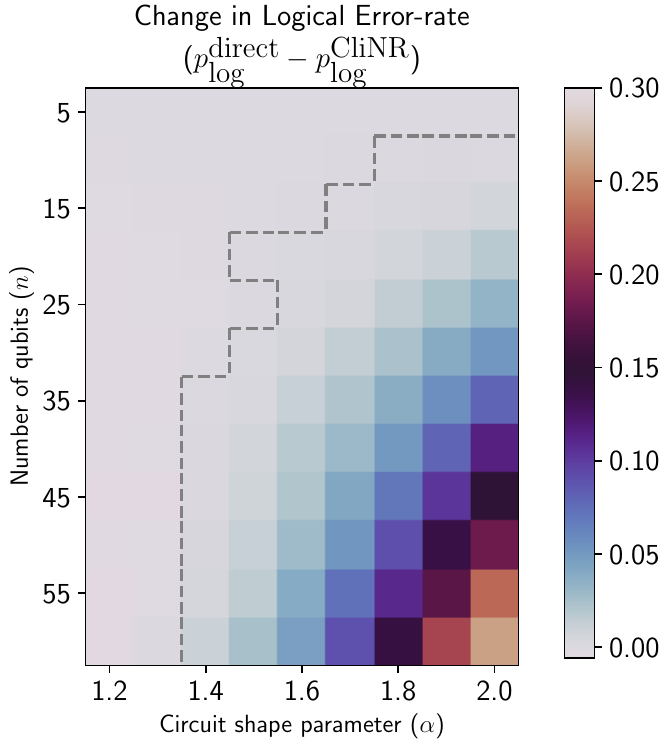}
    \caption{Improvement of the logical error rate of random Clifford circuits under CliNR for the circuit-level noise model of \cref{sec:numerics} with $p_{2}=p_{1}/10=10^{-4}$
    as functions of the qubit count $n$ and the circuit shape $\alpha$ (for circuit size $s=n^{\alpha}$). Gray dashed lines are guides to the eye that delineates between regions where direct implementation is favored (white, upper-left) vs where CliNR yields lower noise (lower-right).
    \label{fig:Grid}}
\end{figure}

\section{CZ noise reduction circuit}
\label{app:cz_noise_reduction}

\begin{figure*}
    \centering
    \includegraphics[width=0.95\linewidth]{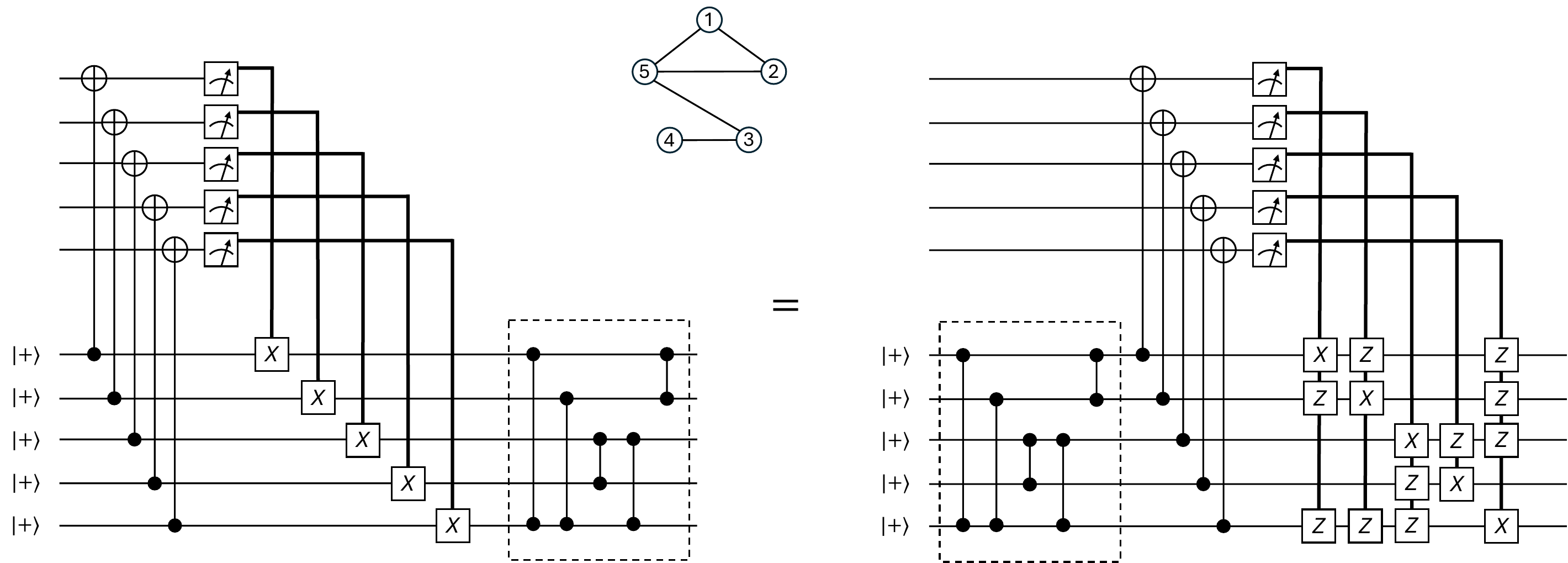}    
    \caption{
    Middle: A graph $G$ with five vertices and five edges. The sequence of $CZ$ gates corresponding to the edges of $G$ defines a unitary $U_G$.
    Left: Teleportation of five qubits, followed by the application of $U_G$ (dashed box) to the teleported qubits.
    Right: The graph state injection circuit associated with $G$.
    Commuting the sequence of $CZ$ gates through the circuit, we see that these two circuits are equivalent.
    }
    \label{fig:graph_state_injection_circuit}
\end{figure*}

Here, we show that the CliNR circuit can be simplified in the case of a Clifford circuit containing only CZ gates.
The basic idea is to implement a sequence of CZ using the one-bit teleportation, shown in \cref{fig:teleportation}(b), instead of the original teleportation circuit.
This results in a circuit that can be seen as a state injection circuit, similar to the $T$-state injection circuit, where the resource state is a graph state instead of a $T$-state.

\subsection{Graph state injection circuit}
\label{subsec:cz_noise_reduction:graph_state_injection}

Any $n$-qubit unitary made with a sequence of $CZ$ gates can be represented by a graph $G$ with $n$ vertices, corresponding to the qubits, and whose edges are in one-to-one correspondence with the $CZ$ gates in the circuit. 
The unitary associated with $G$ is denoted $U_G$.

The {\em graph state injection circuit}, represented in \cref{fig:graph_state_injection_circuit}(right), associated with a graph $G$ with $n$ vertices is the $2n$-qubit circuit defined as follows.
Assume that the vertices of $G$ are indexed by $1,2, \dots, n$ and the qubits are indexed by $1,2,\dots, 2n$.
First, qubits $n+1, n+2, \dots, 2n$ are initialized in the state $\ket +$.
Then, for each edge $\{i, j\}$ in $G$, apply a $CZ$ gate between qubit $n+i$ and qubit $n+j$.
This prepares the so-called {\em graph state}, denoted $\ket G$, on the last $n$ qubits\footnote{For a review of graph states see~\cite{hein2006entanglement}.}.
Then, implement $CNOT$ gates controlled on qubit $n+i$ and targeting qubit $i$ for $i=1,2,\dots, n$.
After the $CNOT$, the first $n$ qubits are measured.
If the outcome of the $i$-th measurement is non-trivial, the Pauli operation
\begin{align}
    X_{n+i} \prod_{j \in N(i)} Z_{n+j}
\end{align}
is applied, where $N(i)$ denotes the set of neighbors of vertex $i$ in $G$.

The following lemma shows that the graph state injection circuit consumes the graph state $\ket G$ prepared on the second block of qubits to implement the operation $U_G$.

\begin{lemma} [Graph state injection circuit]
\label{lemma:graph_state_injection}
The state-injection circuit simultaneously performs the teleportation of the state of the first $n$ qubits onto the last $n$ qubits, and apply the unitary $U_G$ to this state.
\end{lemma}

\begin{proof}
By commuting the sequence of $CZ$ to the end of the graph state injection circuit, we see that this circuit is equivalent to applying the one-bit teleportation circuit to the $n$ input qubits and applying the $CZ$ gates, that is $U_G$, to the output qubits of the teleportations.
This is clear by examination of \cref{fig:graph_state_injection_circuit}.
\end{proof}

\subsection{Noise reduction in CZ circuits}
\label{subsec:cz_noise_reduction:noise_reduction}

\begin{figure}
    \centering
    \includegraphics[width=1\linewidth]{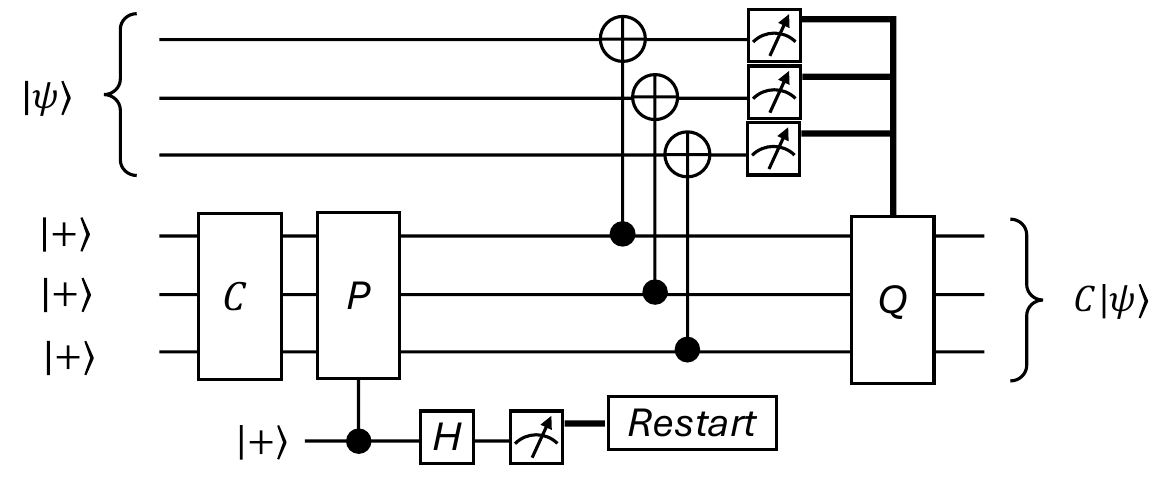}    
    \caption{
    The CZNR circuit implements a $n$-qubit $U$ made with a sequence of $CZ$ using $n+1$ ancilla qubits.
    The ancilla qubits are used to prepare a graph state and check it using $r$ stabilizer measurements.
    In this figure $r=1$ and the stabilizer $P$ of the graph state is measured.
    It is a simplification of the CliNR circuit of \cref{fig:clinr_circuit} in the case of sequences of $CZ$ gates.
    }
    \label{fig:cznr_circuit}
\end{figure}

Similarly to the CliNR circuit, the $CZ$ {\em noise reduction} (CZNR) circuit, described in \cref{fig:cznr_circuit}, is obtained by inserting $r$ stabilizer measurements in the graph state injection circuit, right after the graph state preparation.
Like in the CliNR circuit, the stabilizers are independent stabilizers, selected uniformly at random in the stabilizer group of the first $n$ ancilla qubits and a new set of $r$ stabilizer generators is selected for each run of the circuit.
The circuit restarts if one of the stabilizer measurements returns a non-trivial outcome.
The CZNR circuit is limited to $CZ$ sequences but because it is built from the one-bit teleportation circuit, it consumes only $n$ ancilla qubits instead of $2n$.
We denote by $\CZNR_{t, r}(C)$ the analog of $\CliNR_{t, r}(C)$ for $CZ$ sequences.

The following result is adapted from \cref{theorem:CliNR_s0_r} to $CZ$ sequences, that is Clifford circuits containing only $CZ$ gates.

\begin{figure}
    \centering
    \includegraphics[width=.9\linewidth]{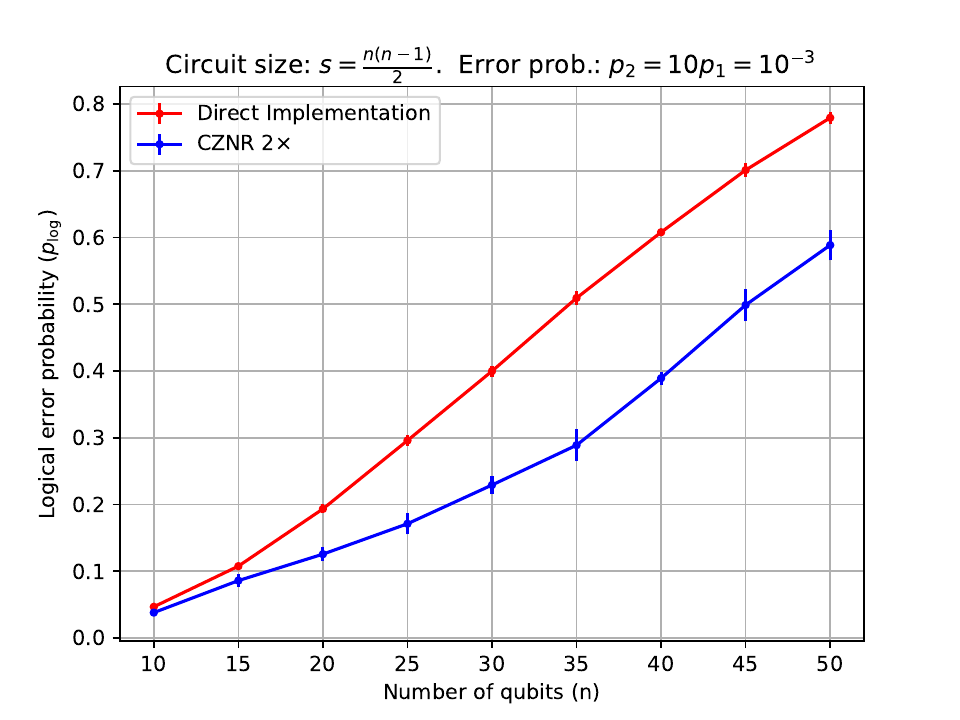}
    \includegraphics[width=.9\linewidth]{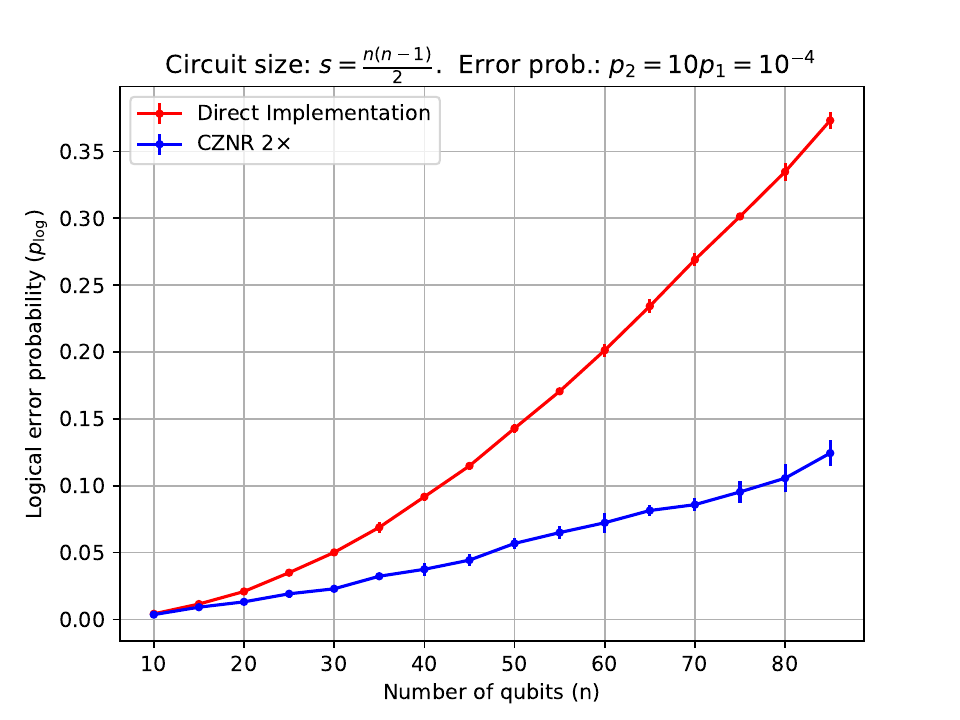}
    \caption{Average logical error rate of a dense block of CZ gates with and without CZNR, in two noise regimes. As in the CliNR case, here $t$ is selected for CZNR$_{t,r}$ such that the overhead is bounded, $\omega_G\leq2$.
    \label{fig:cznr_numerics}}
\end{figure}

\begin{theorem}
\label{theorem:CZNR}
    The circuit $\CZNR_{t, r}(C)$ implements a $CZ$ sequence $C$ with logical error rate
    \begin{align}
        \plog \leq & t \cdot \frac{g_p(n+s) 2^{-r} + 2 g_p(n + 3) + g_p(3n)}{(1-p)^{m_0}}
    \end{align}
    where $m = n + s_0 + (n + 3)r$.
    Moreover, the overhead satisfies 
    $\qubitoverhead = 2 + \frac{1}{n}$ 
    and 
    $\gateoverhead \leq \frac{6n}{s_0} + \frac{2 m_0}{s_0(1-p)^{m}}$.
\end{theorem}

We omit the proof of this result as it is an immediate application of the same technique used to prove \cref{lemma:CliNR} and \cref{theorem:CliNR_s0_r}.

In \cref{fig:cznr_numerics}, we show numerical results comparing direct implementation of a dense ($\mathcal{O} (n^2)$) block of CZ gates against a CZNR equivalent.
As described in the main text for CliNR, we similarly cap the $\omega_G\leq 2$ for CZNR$_{t,r}$ here by varying $t$.

\pagebreak


\begin{thebibliography}{45}%
    \makeatletter
    \providecommand \@ifxundefined [1]{%
     \@ifx{#1\undefined}
    }%
    \providecommand \@ifnum [1]{%
     \ifnum #1\expandafter \@firstoftwo
     \else \expandafter \@secondoftwo
     \fi
    }%
    \providecommand \@ifx [1]{%
     \ifx #1\expandafter \@firstoftwo
     \else \expandafter \@secondoftwo
     \fi
    }%
    \providecommand \natexlab [1]{#1}%
    \providecommand \enquote  [1]{``#1''}%
    \providecommand \bibnamefont  [1]{#1}%
    \providecommand \bibfnamefont [1]{#1}%
    \providecommand \citenamefont [1]{#1}%
    \providecommand \href@noop [0]{\@secondoftwo}%
    \providecommand \href [0]{\begingroup \@sanitize@url \@href}%
    \providecommand \@href[1]{\@@startlink{#1}\@@href}%
    \providecommand \@@href[1]{\endgroup#1\@@endlink}%
    \providecommand \@sanitize@url [0]{\catcode `\\12\catcode `\$12\catcode `\&12\catcode `\#12\catcode `\^12\catcode `\_12\catcode `\%12\relax}%
    \providecommand \@@startlink[1]{}%
    \providecommand \@@endlink[0]{}%
    \providecommand \url  [0]{\begingroup\@sanitize@url \@url }%
    \providecommand \@url [1]{\endgroup\@href {#1}{\urlprefix }}%
    \providecommand \urlprefix  [0]{URL }%
    \providecommand \Eprint [0]{\href }%
    \providecommand \doibase [0]{https://doi.org/}%
    \providecommand \selectlanguage [0]{\@gobble}%
    \providecommand \bibinfo  [0]{\@secondoftwo}%
    \providecommand \bibfield  [0]{\@secondoftwo}%
    \providecommand \translation [1]{[#1]}%
    \providecommand \BibitemOpen [0]{}%
    \providecommand \bibitemStop [0]{}%
    \providecommand \bibitemNoStop [0]{.\EOS\space}%
    \providecommand \EOS [0]{\spacefactor3000\relax}%
    \providecommand \BibitemShut  [1]{\csname bibitem#1\endcsname}%
    \let\auto@bib@innerbib\@empty
    \bibitem [{\citenamefont {Roffe}\ \emph {et~al.}(2018)\citenamefont {Roffe}, \citenamefont {Headley}, \citenamefont {Chancellor}, \citenamefont {Horsman},\ and\ \citenamefont {Kendon}}]{roffe2018protecting}%
      \BibitemOpen
      \bibfield  {author} {\bibinfo {author} {\bibfnamefont {J.}~\bibnamefont {Roffe}}, \bibinfo {author} {\bibfnamefont {D.}~\bibnamefont {Headley}}, \bibinfo {author} {\bibfnamefont {N.}~\bibnamefont {Chancellor}}, \bibinfo {author} {\bibfnamefont {D.}~\bibnamefont {Horsman}},\ and\ \bibinfo {author} {\bibfnamefont {V.}~\bibnamefont {Kendon}},\ }\bibfield  {title} {\bibinfo {title} {Protecting quantum memories using coherent parity check codes},\ }\href@noop {} {\bibfield  {journal} {\bibinfo  {journal} {Quantum Science and Technology}\ }\textbf {\bibinfo {volume} {3}},\ \bibinfo {pages} {035010} (\bibinfo {year} {2018})}\BibitemShut {NoStop}%
    \bibitem [{\citenamefont {Debroy}\ and\ \citenamefont {Brown}(2020)}]{debroy2020extended}%
      \BibitemOpen
      \bibfield  {author} {\bibinfo {author} {\bibfnamefont {D.~M.}\ \bibnamefont {Debroy}}\ and\ \bibinfo {author} {\bibfnamefont {K.~R.}\ \bibnamefont {Brown}},\ }\bibfield  {title} {\bibinfo {title} {Extended flag gadgets for low-overhead circuit verification},\ }\href@noop {} {\bibfield  {journal} {\bibinfo  {journal} {Physical Review A}\ }\textbf {\bibinfo {volume} {102}},\ \bibinfo {pages} {052409} (\bibinfo {year} {2020})}\BibitemShut {NoStop}%
    \bibitem [{\citenamefont {van~den Berg}\ \emph {et~al.}(2023)\citenamefont {van~den Berg}, \citenamefont {Bravyi}, \citenamefont {Gambetta}, \citenamefont {Jurcevic}, \citenamefont {Maslov},\ and\ \citenamefont {Temme}}]{van2023single}%
      \BibitemOpen
      \bibfield  {author} {\bibinfo {author} {\bibfnamefont {E.}~\bibnamefont {van~den Berg}}, \bibinfo {author} {\bibfnamefont {S.}~\bibnamefont {Bravyi}}, \bibinfo {author} {\bibfnamefont {J.~M.}\ \bibnamefont {Gambetta}}, \bibinfo {author} {\bibfnamefont {P.}~\bibnamefont {Jurcevic}}, \bibinfo {author} {\bibfnamefont {D.}~\bibnamefont {Maslov}},\ and\ \bibinfo {author} {\bibfnamefont {K.}~\bibnamefont {Temme}},\ }\bibfield  {title} {\bibinfo {title} {Single-shot error mitigation by coherent pauli checks},\ }\href@noop {} {\bibfield  {journal} {\bibinfo  {journal} {Physical Review Research}\ }\textbf {\bibinfo {volume} {5}},\ \bibinfo {pages} {033193} (\bibinfo {year} {2023})}\BibitemShut {NoStop}%
    \bibitem [{\citenamefont {Shor}(1995)}]{shor1995scheme}%
      \BibitemOpen
      \bibfield  {author} {\bibinfo {author} {\bibfnamefont {P.~W.}\ \bibnamefont {Shor}},\ }\bibfield  {title} {\bibinfo {title} {Scheme for reducing decoherence in quantum computer memory},\ }\href@noop {} {\bibfield  {journal} {\bibinfo  {journal} {Physical review A}\ }\textbf {\bibinfo {volume} {52}},\ \bibinfo {pages} {R2493} (\bibinfo {year} {1995})}\BibitemShut {NoStop}%
    \bibitem [{\citenamefont {Aharonov}\ and\ \citenamefont {Ben-Or}(1997)}]{aharonov1997fault}%
      \BibitemOpen
      \bibfield  {author} {\bibinfo {author} {\bibfnamefont {D.}~\bibnamefont {Aharonov}}\ and\ \bibinfo {author} {\bibfnamefont {M.}~\bibnamefont {Ben-Or}},\ }\bibfield  {title} {\bibinfo {title} {Fault-tolerant quantum computation with constant error},\ }in\ \href@noop {} {\emph {\bibinfo {booktitle} {Proceedings of the twenty-ninth annual ACM symposium on Theory of computing}}}\ (\bibinfo {year} {1997})\ pp.\ \bibinfo {pages} {176--188}\BibitemShut {NoStop}%
    \bibitem [{\citenamefont {Aliferis}\ \emph {et~al.}(2006)\citenamefont {Aliferis}, \citenamefont {Gottesman},\ and\ \citenamefont {Preskill}}]{aliferis2006quantum}%
      \BibitemOpen
      \bibfield  {author} {\bibinfo {author} {\bibfnamefont {P.}~\bibnamefont {Aliferis}}, \bibinfo {author} {\bibfnamefont {D.}~\bibnamefont {Gottesman}},\ and\ \bibinfo {author} {\bibfnamefont {J.}~\bibnamefont {Preskill}},\ }\bibfield  {title} {\bibinfo {title} {Quantum accuracy threshold for concatenated distance-3 codes},\ }\href@noop {} {\bibfield  {journal} {\bibinfo  {journal} {Quantum Info. Comput.}\ }\textbf {\bibinfo {volume} {6}},\ \bibinfo {pages} {97–165} (\bibinfo {year} {2006})}\BibitemShut {NoStop}%
    \bibitem [{\citenamefont {Gidney}\ and\ \citenamefont {Eker{\aa}}(2021)}]{gidney2021factor}%
      \BibitemOpen
      \bibfield  {author} {\bibinfo {author} {\bibfnamefont {C.}~\bibnamefont {Gidney}}\ and\ \bibinfo {author} {\bibfnamefont {M.}~\bibnamefont {Eker{\aa}}},\ }\bibfield  {title} {\bibinfo {title} {How to factor 2048 bit rsa integers in 8 hours using 20 million noisy qubits},\ }\href@noop {} {\bibfield  {journal} {\bibinfo  {journal} {Quantum}\ }\textbf {\bibinfo {volume} {5}},\ \bibinfo {pages} {433} (\bibinfo {year} {2021})}\BibitemShut {NoStop}%
    \bibitem [{\citenamefont {Gottesman}(2013)}]{gottesman2013fault}%
      \BibitemOpen
      \bibfield  {author} {\bibinfo {author} {\bibfnamefont {D.}~\bibnamefont {Gottesman}},\ }\bibfield  {title} {\bibinfo {title} {Fault-tolerant quantum computation with constant overhead},\ }\href@noop {} {\bibfield  {journal} {\bibinfo  {journal} {arXiv preprint arXiv:1310.2984}\ } (\bibinfo {year} {2013})}\BibitemShut {NoStop}%
    \bibitem [{\citenamefont {Breuckmann}\ and\ \citenamefont {Eberhardt}(2021)}]{breuckmann2021quantum}%
      \BibitemOpen
      \bibfield  {author} {\bibinfo {author} {\bibfnamefont {N.~P.}\ \bibnamefont {Breuckmann}}\ and\ \bibinfo {author} {\bibfnamefont {J.~N.}\ \bibnamefont {Eberhardt}},\ }\bibfield  {title} {\bibinfo {title} {Quantum low-density parity-check codes},\ }\href@noop {} {\bibfield  {journal} {\bibinfo  {journal} {PRX Quantum}\ }\textbf {\bibinfo {volume} {2}},\ \bibinfo {pages} {040101} (\bibinfo {year} {2021})}\BibitemShut {NoStop}%
    \bibitem [{\citenamefont {Panteleev}\ and\ \citenamefont {Kalachev}(2022)}]{panteleev2022asymptotically}%
      \BibitemOpen
      \bibfield  {author} {\bibinfo {author} {\bibfnamefont {P.}~\bibnamefont {Panteleev}}\ and\ \bibinfo {author} {\bibfnamefont {G.}~\bibnamefont {Kalachev}},\ }\bibfield  {title} {\bibinfo {title} {Asymptotically good quantum and locally testable classical ldpc codes},\ }in\ \href@noop {} {\emph {\bibinfo {booktitle} {Proceedings of the 54th Annual ACM SIGACT Symposium on Theory of Computing}}}\ (\bibinfo {year} {2022})\ pp.\ \bibinfo {pages} {375--388}\BibitemShut {NoStop}%
    \bibitem [{\citenamefont {Leverrier}\ and\ \citenamefont {Z{\'e}mor}(2022)}]{leverrier2022quantum}%
      \BibitemOpen
      \bibfield  {author} {\bibinfo {author} {\bibfnamefont {A.}~\bibnamefont {Leverrier}}\ and\ \bibinfo {author} {\bibfnamefont {G.}~\bibnamefont {Z{\'e}mor}},\ }\bibfield  {title} {\bibinfo {title} {Quantum tanner codes},\ }in\ \href@noop {} {\emph {\bibinfo {booktitle} {2022 IEEE 63rd Annual Symposium on Foundations of Computer Science (FOCS)}}}\ (\bibinfo {organization} {IEEE},\ \bibinfo {year} {2022})\ pp.\ \bibinfo {pages} {872--883}\BibitemShut {NoStop}%
    \bibitem [{\citenamefont {Pattison}\ \emph {et~al.}(2023)\citenamefont {Pattison}, \citenamefont {Krishna},\ and\ \citenamefont {Preskill}}]{pattison2023hierarchical}%
      \BibitemOpen
      \bibfield  {author} {\bibinfo {author} {\bibfnamefont {C.~A.}\ \bibnamefont {Pattison}}, \bibinfo {author} {\bibfnamefont {A.}~\bibnamefont {Krishna}},\ and\ \bibinfo {author} {\bibfnamefont {J.}~\bibnamefont {Preskill}},\ }\bibfield  {title} {\bibinfo {title} {Hierarchical memories: Simulating quantum ldpc codes with local gates},\ }\href@noop {} {\bibfield  {journal} {\bibinfo  {journal} {arXiv preprint arXiv:2303.04798}\ } (\bibinfo {year} {2023})}\BibitemShut {NoStop}%
    \bibitem [{\citenamefont {Yoshida}\ \emph {et~al.}(2024)\citenamefont {Yoshida}, \citenamefont {Tamiya},\ and\ \citenamefont {Yamasaki}}]{yoshida2024concatenate}%
      \BibitemOpen
      \bibfield  {author} {\bibinfo {author} {\bibfnamefont {S.}~\bibnamefont {Yoshida}}, \bibinfo {author} {\bibfnamefont {S.}~\bibnamefont {Tamiya}},\ and\ \bibinfo {author} {\bibfnamefont {H.}~\bibnamefont {Yamasaki}},\ }\bibfield  {title} {\bibinfo {title} {Concatenate codes, save qubits},\ }\href@noop {} {\bibfield  {journal} {\bibinfo  {journal} {arXiv preprint arXiv:2402.09606}\ } (\bibinfo {year} {2024})}\BibitemShut {NoStop}%
    \bibitem [{\citenamefont {Yamasaki}\ and\ \citenamefont {Koashi}(2024)}]{yamasaki2024time}%
      \BibitemOpen
      \bibfield  {author} {\bibinfo {author} {\bibfnamefont {H.}~\bibnamefont {Yamasaki}}\ and\ \bibinfo {author} {\bibfnamefont {M.}~\bibnamefont {Koashi}},\ }\bibfield  {title} {\bibinfo {title} {Time-efficient constant-space-overhead fault-tolerant quantum computation},\ }\href@noop {} {\bibfield  {journal} {\bibinfo  {journal} {Nature Physics}\ ,\ \bibinfo {pages} {1}} (\bibinfo {year} {2024})}\BibitemShut {NoStop}%
    \bibitem [{\citenamefont {Zhou}\ \emph {et~al.}(2024)\citenamefont {Zhou}, \citenamefont {Zhao}, \citenamefont {Cain}, \citenamefont {Bluvstein}, \citenamefont {Duckering}, \citenamefont {Hu}, \citenamefont {Wang}, \citenamefont {Kubica},\ and\ \citenamefont {Lukin}}]{zhou2024algorithmic}%
      \BibitemOpen
      \bibfield  {author} {\bibinfo {author} {\bibfnamefont {H.}~\bibnamefont {Zhou}}, \bibinfo {author} {\bibfnamefont {C.}~\bibnamefont {Zhao}}, \bibinfo {author} {\bibfnamefont {M.}~\bibnamefont {Cain}}, \bibinfo {author} {\bibfnamefont {D.}~\bibnamefont {Bluvstein}}, \bibinfo {author} {\bibfnamefont {C.}~\bibnamefont {Duckering}}, \bibinfo {author} {\bibfnamefont {H.-Y.}\ \bibnamefont {Hu}}, \bibinfo {author} {\bibfnamefont {S.-T.}\ \bibnamefont {Wang}}, \bibinfo {author} {\bibfnamefont {A.}~\bibnamefont {Kubica}},\ and\ \bibinfo {author} {\bibfnamefont {M.~D.}\ \bibnamefont {Lukin}},\ }\bibfield  {title} {\bibinfo {title} {Algorithmic fault tolerance for fast quantum computing},\ }\href@noop {} {\bibfield  {journal} {\bibinfo  {journal} {arXiv preprint arXiv:2406.17653}\ } (\bibinfo {year} {2024})}\BibitemShut {NoStop}%
    \bibitem [{\citenamefont {Tremblay}\ \emph {et~al.}(2022)\citenamefont {Tremblay}, \citenamefont {Delfosse},\ and\ \citenamefont {Beverland}}]{tremblay2022constant}%
      \BibitemOpen
      \bibfield  {author} {\bibinfo {author} {\bibfnamefont {M.~A.}\ \bibnamefont {Tremblay}}, \bibinfo {author} {\bibfnamefont {N.}~\bibnamefont {Delfosse}},\ and\ \bibinfo {author} {\bibfnamefont {M.~E.}\ \bibnamefont {Beverland}},\ }\bibfield  {title} {\bibinfo {title} {Constant-overhead quantum error correction with thin planar connectivity},\ }\href@noop {} {\bibfield  {journal} {\bibinfo  {journal} {Physical Review Letters}\ }\textbf {\bibinfo {volume} {129}},\ \bibinfo {pages} {050504} (\bibinfo {year} {2022})}\BibitemShut {NoStop}%
    \bibitem [{\citenamefont {Higgott}\ and\ \citenamefont {Breuckmann}(2023)}]{higgott2023constructions}%
      \BibitemOpen
      \bibfield  {author} {\bibinfo {author} {\bibfnamefont {O.}~\bibnamefont {Higgott}}\ and\ \bibinfo {author} {\bibfnamefont {N.~P.}\ \bibnamefont {Breuckmann}},\ }\bibfield  {title} {\bibinfo {title} {Constructions and performance of hyperbolic and semi-hyperbolic floquet codes},\ }\href@noop {} {\bibfield  {journal} {\bibinfo  {journal} {arXiv preprint arXiv:2308.03750}\ } (\bibinfo {year} {2023})}\BibitemShut {NoStop}%
    \bibitem [{\citenamefont {Bravyi}\ \emph {et~al.}(2024)\citenamefont {Bravyi}, \citenamefont {Cross}, \citenamefont {Gambetta}, \citenamefont {Maslov}, \citenamefont {Rall},\ and\ \citenamefont {Yoder}}]{bravyi2024high}%
      \BibitemOpen
      \bibfield  {author} {\bibinfo {author} {\bibfnamefont {S.}~\bibnamefont {Bravyi}}, \bibinfo {author} {\bibfnamefont {A.~W.}\ \bibnamefont {Cross}}, \bibinfo {author} {\bibfnamefont {J.~M.}\ \bibnamefont {Gambetta}}, \bibinfo {author} {\bibfnamefont {D.}~\bibnamefont {Maslov}}, \bibinfo {author} {\bibfnamefont {P.}~\bibnamefont {Rall}},\ and\ \bibinfo {author} {\bibfnamefont {T.~J.}\ \bibnamefont {Yoder}},\ }\bibfield  {title} {\bibinfo {title} {High-threshold and low-overhead fault-tolerant quantum memory},\ }\href@noop {} {\bibfield  {journal} {\bibinfo  {journal} {Nature}\ }\textbf {\bibinfo {volume} {627}},\ \bibinfo {pages} {778} (\bibinfo {year} {2024})}\BibitemShut {NoStop}%
    \bibitem [{\citenamefont {Scruby}\ \emph {et~al.}(2024)\citenamefont {Scruby}, \citenamefont {Hillmann},\ and\ \citenamefont {Roffe}}]{scruby2024high}%
      \BibitemOpen
      \bibfield  {author} {\bibinfo {author} {\bibfnamefont {T.~R.}\ \bibnamefont {Scruby}}, \bibinfo {author} {\bibfnamefont {T.}~\bibnamefont {Hillmann}},\ and\ \bibinfo {author} {\bibfnamefont {J.}~\bibnamefont {Roffe}},\ }\bibfield  {title} {\bibinfo {title} {High-threshold, low-overhead and single-shot decodable fault-tolerant quantum memory},\ }\href@noop {} {\bibfield  {journal} {\bibinfo  {journal} {arXiv preprint arXiv:2406.14445}\ } (\bibinfo {year} {2024})}\BibitemShut {NoStop}%
    \bibitem [{\citenamefont {Egan}\ \emph {et~al.}(2021)\citenamefont {Egan}, \citenamefont {Debroy}, \citenamefont {Noel}, \citenamefont {Risinger}, \citenamefont {Zhu}, \citenamefont {Biswas}, \citenamefont {Newman}, \citenamefont {Li}, \citenamefont {Brown}, \citenamefont {Cetina} \emph {et~al.}}]{egan2021fault}%
      \BibitemOpen
      \bibfield  {author} {\bibinfo {author} {\bibfnamefont {L.}~\bibnamefont {Egan}}, \bibinfo {author} {\bibfnamefont {D.~M.}\ \bibnamefont {Debroy}}, \bibinfo {author} {\bibfnamefont {C.}~\bibnamefont {Noel}}, \bibinfo {author} {\bibfnamefont {A.}~\bibnamefont {Risinger}}, \bibinfo {author} {\bibfnamefont {D.}~\bibnamefont {Zhu}}, \bibinfo {author} {\bibfnamefont {D.}~\bibnamefont {Biswas}}, \bibinfo {author} {\bibfnamefont {M.}~\bibnamefont {Newman}}, \bibinfo {author} {\bibfnamefont {M.}~\bibnamefont {Li}}, \bibinfo {author} {\bibfnamefont {K.~R.}\ \bibnamefont {Brown}}, \bibinfo {author} {\bibfnamefont {M.}~\bibnamefont {Cetina}}, \emph {et~al.},\ }\bibfield  {title} {\bibinfo {title} {Fault-tolerant control of an error-corrected qubit},\ }\href@noop {} {\bibfield  {journal} {\bibinfo  {journal} {Nature}\ }\textbf {\bibinfo {volume} {598}},\ \bibinfo {pages} {281} (\bibinfo {year} {2021})}\BibitemShut {NoStop}%
    \bibitem [{\citenamefont {Ryan-Anderson}\ \emph {et~al.}(2021)\citenamefont {Ryan-Anderson}, \citenamefont {Bohnet}, \citenamefont {Lee}, \citenamefont {Gresh}, \citenamefont {Hankin}, \citenamefont {Gaebler}, \citenamefont {Francois}, \citenamefont {Chernoguzov}, \citenamefont {Lucchetti}, \citenamefont {Brown} \emph {et~al.}}]{ryan2021realization}%
      \BibitemOpen
      \bibfield  {author} {\bibinfo {author} {\bibfnamefont {C.}~\bibnamefont {Ryan-Anderson}}, \bibinfo {author} {\bibfnamefont {J.~G.}\ \bibnamefont {Bohnet}}, \bibinfo {author} {\bibfnamefont {K.}~\bibnamefont {Lee}}, \bibinfo {author} {\bibfnamefont {D.}~\bibnamefont {Gresh}}, \bibinfo {author} {\bibfnamefont {A.}~\bibnamefont {Hankin}}, \bibinfo {author} {\bibfnamefont {J.}~\bibnamefont {Gaebler}}, \bibinfo {author} {\bibfnamefont {D.}~\bibnamefont {Francois}}, \bibinfo {author} {\bibfnamefont {A.}~\bibnamefont {Chernoguzov}}, \bibinfo {author} {\bibfnamefont {D.}~\bibnamefont {Lucchetti}}, \bibinfo {author} {\bibfnamefont {N.~C.}\ \bibnamefont {Brown}}, \emph {et~al.},\ }\bibfield  {title} {\bibinfo {title} {Realization of real-time fault-tolerant quantum error correction},\ }\href@noop {} {\bibfield  {journal} {\bibinfo  {journal} {Physical Review X}\ }\textbf {\bibinfo {volume} {11}},\ \bibinfo {pages} {041058} (\bibinfo {year} {2021})}\BibitemShut {NoStop}%
    \bibitem [{\citenamefont {Postler}\ \emph {et~al.}(2023)\citenamefont {Postler}, \citenamefont {Butt}, \citenamefont {Pogorelov}, \citenamefont {Marciniak}, \citenamefont {Heu{\ss}en}, \citenamefont {Blatt}, \citenamefont {Schindler}, \citenamefont {Rispler}, \citenamefont {M{\"u}ller},\ and\ \citenamefont {Monz}}]{postler2023demonstration}%
      \BibitemOpen
      \bibfield  {author} {\bibinfo {author} {\bibfnamefont {L.}~\bibnamefont {Postler}}, \bibinfo {author} {\bibfnamefont {F.}~\bibnamefont {Butt}}, \bibinfo {author} {\bibfnamefont {I.}~\bibnamefont {Pogorelov}}, \bibinfo {author} {\bibfnamefont {C.~D.}\ \bibnamefont {Marciniak}}, \bibinfo {author} {\bibfnamefont {S.}~\bibnamefont {Heu{\ss}en}}, \bibinfo {author} {\bibfnamefont {R.}~\bibnamefont {Blatt}}, \bibinfo {author} {\bibfnamefont {P.}~\bibnamefont {Schindler}}, \bibinfo {author} {\bibfnamefont {M.}~\bibnamefont {Rispler}}, \bibinfo {author} {\bibfnamefont {M.}~\bibnamefont {M{\"u}ller}},\ and\ \bibinfo {author} {\bibfnamefont {T.}~\bibnamefont {Monz}},\ }\bibfield  {title} {\bibinfo {title} {Demonstration of fault-tolerant steane quantum error correction},\ }\href@noop {} {\bibfield  {journal} {\bibinfo  {journal} {arXiv preprint arXiv:2312.09745}\ } (\bibinfo {year} {2023})}\BibitemShut {NoStop}%
    \bibitem [{\citenamefont {Abobeih}\ \emph {et~al.}(2022)\citenamefont {Abobeih}, \citenamefont {Wang}, \citenamefont {Randall}, \citenamefont {Loenen}, \citenamefont {Bradley}, \citenamefont {Markham}, \citenamefont {Twitchen}, \citenamefont {Terhal},\ and\ \citenamefont {Taminiau}}]{abobeih2022fault}%
      \BibitemOpen
      \bibfield  {author} {\bibinfo {author} {\bibfnamefont {M.~H.}\ \bibnamefont {Abobeih}}, \bibinfo {author} {\bibfnamefont {Y.}~\bibnamefont {Wang}}, \bibinfo {author} {\bibfnamefont {J.}~\bibnamefont {Randall}}, \bibinfo {author} {\bibfnamefont {S.}~\bibnamefont {Loenen}}, \bibinfo {author} {\bibfnamefont {C.~E.}\ \bibnamefont {Bradley}}, \bibinfo {author} {\bibfnamefont {M.}~\bibnamefont {Markham}}, \bibinfo {author} {\bibfnamefont {D.~J.}\ \bibnamefont {Twitchen}}, \bibinfo {author} {\bibfnamefont {B.~M.}\ \bibnamefont {Terhal}},\ and\ \bibinfo {author} {\bibfnamefont {T.~H.}\ \bibnamefont {Taminiau}},\ }\bibfield  {title} {\bibinfo {title} {Fault-tolerant operation of a logical qubit in a diamond quantum processor},\ }\href@noop {} {\bibfield  {journal} {\bibinfo  {journal} {Nature}\ }\textbf {\bibinfo {volume} {606}},\ \bibinfo {pages} {884} (\bibinfo {year} {2022})}\BibitemShut {NoStop}%
    \bibitem [{\citenamefont {Krinner}\ \emph {et~al.}(2022)\citenamefont {Krinner}, \citenamefont {Lacroix}, \citenamefont {Remm}, \citenamefont {Di~Paolo}, \citenamefont {Genois}, \citenamefont {Leroux}, \citenamefont {Hellings}, \citenamefont {Lazar}, \citenamefont {Swiadek}, \citenamefont {Herrmann} \emph {et~al.}}]{krinner2022realizing}%
      \BibitemOpen
      \bibfield  {author} {\bibinfo {author} {\bibfnamefont {S.}~\bibnamefont {Krinner}}, \bibinfo {author} {\bibfnamefont {N.}~\bibnamefont {Lacroix}}, \bibinfo {author} {\bibfnamefont {A.}~\bibnamefont {Remm}}, \bibinfo {author} {\bibfnamefont {A.}~\bibnamefont {Di~Paolo}}, \bibinfo {author} {\bibfnamefont {E.}~\bibnamefont {Genois}}, \bibinfo {author} {\bibfnamefont {C.}~\bibnamefont {Leroux}}, \bibinfo {author} {\bibfnamefont {C.}~\bibnamefont {Hellings}}, \bibinfo {author} {\bibfnamefont {S.}~\bibnamefont {Lazar}}, \bibinfo {author} {\bibfnamefont {F.}~\bibnamefont {Swiadek}}, \bibinfo {author} {\bibfnamefont {J.}~\bibnamefont {Herrmann}}, \emph {et~al.},\ }\bibfield  {title} {\bibinfo {title} {Realizing repeated quantum error correction in a distance-three surface code},\ }\href@noop {} {\bibfield  {journal} {\bibinfo  {journal} {Nature}\ }\textbf {\bibinfo {volume} {605}},\ \bibinfo {pages} {669} (\bibinfo {year} {2022})}\BibitemShut {NoStop}%
    \bibitem [{\citenamefont {Zhao}\ \emph {et~al.}(2022)\citenamefont {Zhao}, \citenamefont {Ye}, \citenamefont {Huang}, \citenamefont {Zhang}, \citenamefont {Wu}, \citenamefont {Guan}, \citenamefont {Zhu}, \citenamefont {Wei}, \citenamefont {He}, \citenamefont {Cao} \emph {et~al.}}]{zhao2022realization}%
      \BibitemOpen
      \bibfield  {author} {\bibinfo {author} {\bibfnamefont {Y.}~\bibnamefont {Zhao}}, \bibinfo {author} {\bibfnamefont {Y.}~\bibnamefont {Ye}}, \bibinfo {author} {\bibfnamefont {H.-L.}\ \bibnamefont {Huang}}, \bibinfo {author} {\bibfnamefont {Y.}~\bibnamefont {Zhang}}, \bibinfo {author} {\bibfnamefont {D.}~\bibnamefont {Wu}}, \bibinfo {author} {\bibfnamefont {H.}~\bibnamefont {Guan}}, \bibinfo {author} {\bibfnamefont {Q.}~\bibnamefont {Zhu}}, \bibinfo {author} {\bibfnamefont {Z.}~\bibnamefont {Wei}}, \bibinfo {author} {\bibfnamefont {T.}~\bibnamefont {He}}, \bibinfo {author} {\bibfnamefont {S.}~\bibnamefont {Cao}}, \emph {et~al.},\ }\bibfield  {title} {\bibinfo {title} {Realization of an error-correcting surface code with superconducting qubits},\ }\href@noop {} {\bibfield  {journal} {\bibinfo  {journal} {Physical Review Letters}\ }\textbf {\bibinfo {volume} {129}},\ \bibinfo {pages} {030501} (\bibinfo {year} {2022})}\BibitemShut {NoStop}%
    \bibitem [{goo(2023)}]{google2023suppressing}%
      \BibitemOpen
      \bibfield  {title} {\bibinfo {title} {Suppressing quantum errors by scaling a surface code logical qubit},\ }\href@noop {} {\bibfield  {journal} {\bibinfo  {journal} {Nature}\ }\textbf {\bibinfo {volume} {614}},\ \bibinfo {pages} {676} (\bibinfo {year} {2023})}\BibitemShut {NoStop}%
    \bibitem [{\citenamefont {Bluvstein}\ \emph {et~al.}(2024)\citenamefont {Bluvstein}, \citenamefont {Evered}, \citenamefont {Geim}, \citenamefont {Li}, \citenamefont {Zhou}, \citenamefont {Manovitz}, \citenamefont {Ebadi}, \citenamefont {Cain}, \citenamefont {Kalinowski}, \citenamefont {Hangleiter} \emph {et~al.}}]{bluvstein2024logical}%
      \BibitemOpen
      \bibfield  {author} {\bibinfo {author} {\bibfnamefont {D.}~\bibnamefont {Bluvstein}}, \bibinfo {author} {\bibfnamefont {S.~J.}\ \bibnamefont {Evered}}, \bibinfo {author} {\bibfnamefont {A.~A.}\ \bibnamefont {Geim}}, \bibinfo {author} {\bibfnamefont {S.~H.}\ \bibnamefont {Li}}, \bibinfo {author} {\bibfnamefont {H.}~\bibnamefont {Zhou}}, \bibinfo {author} {\bibfnamefont {T.}~\bibnamefont {Manovitz}}, \bibinfo {author} {\bibfnamefont {S.}~\bibnamefont {Ebadi}}, \bibinfo {author} {\bibfnamefont {M.}~\bibnamefont {Cain}}, \bibinfo {author} {\bibfnamefont {M.}~\bibnamefont {Kalinowski}}, \bibinfo {author} {\bibfnamefont {D.}~\bibnamefont {Hangleiter}}, \emph {et~al.},\ }\bibfield  {title} {\bibinfo {title} {Logical quantum processor based on reconfigurable atom arrays},\ }\href@noop {} {\bibfield  {journal} {\bibinfo  {journal} {Nature}\ }\textbf {\bibinfo {volume} {626}},\ \bibinfo {pages} {58} (\bibinfo {year} {2024})}\BibitemShut {NoStop}%
    \bibitem [{\citenamefont {Ryan-Anderson}\ \emph {et~al.}(2024)\citenamefont {Ryan-Anderson}, \citenamefont {Brown}, \citenamefont {Baldwin}, \citenamefont {Dreiling}, \citenamefont {Foltz}, \citenamefont {Gaebler}, \citenamefont {Gatterman}, \citenamefont {Hewitt}, \citenamefont {Holliman}, \citenamefont {Horst} \emph {et~al.}}]{ryan2024high}%
      \BibitemOpen
      \bibfield  {author} {\bibinfo {author} {\bibfnamefont {C.}~\bibnamefont {Ryan-Anderson}}, \bibinfo {author} {\bibfnamefont {N.}~\bibnamefont {Brown}}, \bibinfo {author} {\bibfnamefont {C.}~\bibnamefont {Baldwin}}, \bibinfo {author} {\bibfnamefont {J.}~\bibnamefont {Dreiling}}, \bibinfo {author} {\bibfnamefont {C.}~\bibnamefont {Foltz}}, \bibinfo {author} {\bibfnamefont {J.}~\bibnamefont {Gaebler}}, \bibinfo {author} {\bibfnamefont {T.}~\bibnamefont {Gatterman}}, \bibinfo {author} {\bibfnamefont {N.}~\bibnamefont {Hewitt}}, \bibinfo {author} {\bibfnamefont {C.}~\bibnamefont {Holliman}}, \bibinfo {author} {\bibfnamefont {C.}~\bibnamefont {Horst}}, \emph {et~al.},\ }\bibfield  {title} {\bibinfo {title} {High-fidelity and fault-tolerant teleportation of a logical qubit using transversal gates and lattice surgery on a trapped-ion quantum computer},\ }\href@noop {} {\bibfield  {journal} {\bibinfo  {journal} {arXiv preprint arXiv:2404.16728}\ } (\bibinfo {year} {2024})}\BibitemShut {NoStop}%
    \bibitem [{\citenamefont {Mayer}\ \emph {et~al.}(2024)\citenamefont {Mayer}, \citenamefont {Ryan-Anderson}, \citenamefont {Brown}, \citenamefont {Durso-Sabina}, \citenamefont {Baldwin}, \citenamefont {Hayes}, \citenamefont {Dreiling}, \citenamefont {Foltz}, \citenamefont {Gaebler}, \citenamefont {Gatterman} \emph {et~al.}}]{mayer2024benchmarking}%
      \BibitemOpen
      \bibfield  {author} {\bibinfo {author} {\bibfnamefont {K.}~\bibnamefont {Mayer}}, \bibinfo {author} {\bibfnamefont {C.}~\bibnamefont {Ryan-Anderson}}, \bibinfo {author} {\bibfnamefont {N.}~\bibnamefont {Brown}}, \bibinfo {author} {\bibfnamefont {E.}~\bibnamefont {Durso-Sabina}}, \bibinfo {author} {\bibfnamefont {C.~H.}\ \bibnamefont {Baldwin}}, \bibinfo {author} {\bibfnamefont {D.}~\bibnamefont {Hayes}}, \bibinfo {author} {\bibfnamefont {J.~M.}\ \bibnamefont {Dreiling}}, \bibinfo {author} {\bibfnamefont {C.}~\bibnamefont {Foltz}}, \bibinfo {author} {\bibfnamefont {J.~P.}\ \bibnamefont {Gaebler}}, \bibinfo {author} {\bibfnamefont {T.~M.}\ \bibnamefont {Gatterman}}, \emph {et~al.},\ }\bibfield  {title} {\bibinfo {title} {Benchmarking logical three-qubit quantum fourier transform encoded in the steane code on a trapped-ion quantum computer},\ }\href@noop {} {\bibfield  {journal} {\bibinfo  {journal} {arXiv preprint arXiv:2404.08616}\ } (\bibinfo {year} {2024})}\BibitemShut {NoStop}%
    \bibitem [{\citenamefont {Li}\ and\ \citenamefont {Benjamin}(2017)}]{li2017efficient}%
      \BibitemOpen
      \bibfield  {author} {\bibinfo {author} {\bibfnamefont {Y.}~\bibnamefont {Li}}\ and\ \bibinfo {author} {\bibfnamefont {S.~C.}\ \bibnamefont {Benjamin}},\ }\bibfield  {title} {\bibinfo {title} {Efficient variational quantum simulator incorporating active error minimization},\ }\href@noop {} {\bibfield  {journal} {\bibinfo  {journal} {Physical Review X}\ }\textbf {\bibinfo {volume} {7}},\ \bibinfo {pages} {021050} (\bibinfo {year} {2017})}\BibitemShut {NoStop}%
    \bibitem [{\citenamefont {Temme}\ \emph {et~al.}(2017)\citenamefont {Temme}, \citenamefont {Bravyi},\ and\ \citenamefont {Gambetta}}]{temme2017error}%
      \BibitemOpen
      \bibfield  {author} {\bibinfo {author} {\bibfnamefont {K.}~\bibnamefont {Temme}}, \bibinfo {author} {\bibfnamefont {S.}~\bibnamefont {Bravyi}},\ and\ \bibinfo {author} {\bibfnamefont {J.~M.}\ \bibnamefont {Gambetta}},\ }\bibfield  {title} {\bibinfo {title} {Error mitigation for short-depth quantum circuits},\ }\href@noop {} {\bibfield  {journal} {\bibinfo  {journal} {Physical review letters}\ }\textbf {\bibinfo {volume} {119}},\ \bibinfo {pages} {180509} (\bibinfo {year} {2017})}\BibitemShut {NoStop}%
    \bibitem [{\citenamefont {Cai}\ \emph {et~al.}(2023)\citenamefont {Cai}, \citenamefont {Babbush}, \citenamefont {Benjamin}, \citenamefont {Endo}, \citenamefont {Huggins}, \citenamefont {Li}, \citenamefont {McClean},\ and\ \citenamefont {O’Brien}}]{cai2023quantum}%
      \BibitemOpen
      \bibfield  {author} {\bibinfo {author} {\bibfnamefont {Z.}~\bibnamefont {Cai}}, \bibinfo {author} {\bibfnamefont {R.}~\bibnamefont {Babbush}}, \bibinfo {author} {\bibfnamefont {S.~C.}\ \bibnamefont {Benjamin}}, \bibinfo {author} {\bibfnamefont {S.}~\bibnamefont {Endo}}, \bibinfo {author} {\bibfnamefont {W.~J.}\ \bibnamefont {Huggins}}, \bibinfo {author} {\bibfnamefont {Y.}~\bibnamefont {Li}}, \bibinfo {author} {\bibfnamefont {J.~R.}\ \bibnamefont {McClean}},\ and\ \bibinfo {author} {\bibfnamefont {T.~E.}\ \bibnamefont {O’Brien}},\ }\bibfield  {title} {\bibinfo {title} {Quantum error mitigation},\ }\href@noop {} {\bibfield  {journal} {\bibinfo  {journal} {Reviews of Modern Physics}\ }\textbf {\bibinfo {volume} {95}},\ \bibinfo {pages} {045005} (\bibinfo {year} {2023})}\BibitemShut {NoStop}%
    \bibitem [{\citenamefont {Gottesman}\ and\ \citenamefont {Chuang}(1999)}]{gottesman1999quantum}%
      \BibitemOpen
      \bibfield  {author} {\bibinfo {author} {\bibfnamefont {D.}~\bibnamefont {Gottesman}}\ and\ \bibinfo {author} {\bibfnamefont {I.~L.}\ \bibnamefont {Chuang}},\ }\bibfield  {title} {\bibinfo {title} {Quantum teleportation is a universal computational primitive},\ }\href@noop {} {\bibfield  {journal} {\bibinfo  {journal} {arXiv preprint quant-ph/9908010}\ } (\bibinfo {year} {1999})}\BibitemShut {NoStop}%
    \bibitem [{\citenamefont {Bravyi}\ and\ \citenamefont {Kitaev}(2005)}]{bravyi2005universal}%
      \BibitemOpen
      \bibfield  {author} {\bibinfo {author} {\bibfnamefont {S.}~\bibnamefont {Bravyi}}\ and\ \bibinfo {author} {\bibfnamefont {A.}~\bibnamefont {Kitaev}},\ }\bibfield  {title} {\bibinfo {title} {Universal quantum computation with ideal clifford gates and noisy ancillas},\ }\href@noop {} {\bibfield  {journal} {\bibinfo  {journal} {Physical Review A}\ }\textbf {\bibinfo {volume} {71}},\ \bibinfo {pages} {022316} (\bibinfo {year} {2005})}\BibitemShut {NoStop}%
    \bibitem [{\citenamefont {Gonzales}\ \emph {et~al.}(2023)\citenamefont {Gonzales}, \citenamefont {Shaydulin}, \citenamefont {Saleem},\ and\ \citenamefont {Suchara}}]{gonzales2023quantum}%
      \BibitemOpen
      \bibfield  {author} {\bibinfo {author} {\bibfnamefont {A.}~\bibnamefont {Gonzales}}, \bibinfo {author} {\bibfnamefont {R.}~\bibnamefont {Shaydulin}}, \bibinfo {author} {\bibfnamefont {Z.~H.}\ \bibnamefont {Saleem}},\ and\ \bibinfo {author} {\bibfnamefont {M.}~\bibnamefont {Suchara}},\ }\bibfield  {title} {\bibinfo {title} {Quantum error mitigation by pauli check sandwiching},\ }\href@noop {} {\bibfield  {journal} {\bibinfo  {journal} {Scientific Reports}\ }\textbf {\bibinfo {volume} {13}},\ \bibinfo {pages} {2122} (\bibinfo {year} {2023})}\BibitemShut {NoStop}%
    \bibitem [{\citenamefont {Anand}\ and\ \citenamefont {Brown}(2023)}]{anand2023hamiltonians}%
      \BibitemOpen
      \bibfield  {author} {\bibinfo {author} {\bibfnamefont {A.}~\bibnamefont {Anand}}\ and\ \bibinfo {author} {\bibfnamefont {K.~R.}\ \bibnamefont {Brown}},\ }\bibfield  {title} {\bibinfo {title} {Hamiltonians, groups, graphs and ans$\backslash$" atze},\ }\href@noop {} {\bibfield  {journal} {\bibinfo  {journal} {arXiv preprint arXiv:2312.17146}\ } (\bibinfo {year} {2023})}\BibitemShut {NoStop}%
    \bibitem [{\citenamefont {Gottesman}(1998)}]{gottesman1998heisenberg}%
      \BibitemOpen
      \bibfield  {author} {\bibinfo {author} {\bibfnamefont {D.}~\bibnamefont {Gottesman}},\ }\bibfield  {title} {\bibinfo {title} {The heisenberg representation of quantum computers},\ }\href@noop {} {\bibfield  {journal} {\bibinfo  {journal} {arXiv preprint quant-ph/9807006}\ } (\bibinfo {year} {1998})}\BibitemShut {NoStop}%
    \bibitem [{\citenamefont {Aaronson}\ and\ \citenamefont {Gottesman}(2004)}]{aaronson2004improved}%
      \BibitemOpen
      \bibfield  {author} {\bibinfo {author} {\bibfnamefont {S.}~\bibnamefont {Aaronson}}\ and\ \bibinfo {author} {\bibfnamefont {D.}~\bibnamefont {Gottesman}},\ }\bibfield  {title} {\bibinfo {title} {Improved simulation of stabilizer circuits},\ }\href@noop {} {\bibfield  {journal} {\bibinfo  {journal} {Physical Review A}\ }\textbf {\bibinfo {volume} {70}},\ \bibinfo {pages} {052328} (\bibinfo {year} {2004})}\BibitemShut {NoStop}%
    \bibitem [{\citenamefont {Bravyi}\ and\ \citenamefont {Maslov}(2021)}]{bravyimaslov2021}%
      \BibitemOpen
      \bibfield  {author} {\bibinfo {author} {\bibfnamefont {S.}~\bibnamefont {Bravyi}}\ and\ \bibinfo {author} {\bibfnamefont {D.}~\bibnamefont {Maslov}},\ }\bibfield  {title} {\bibinfo {title} {Hadamard-free circuits expose the structure of the clifford group},\ }\href {https://doi.org/10.1109/TIT.2021.3081415} {\bibfield  {journal} {\bibinfo  {journal} {IEEE Transactions on Information Theory}\ }\textbf {\bibinfo {volume} {67}},\ \bibinfo {pages} {4546} (\bibinfo {year} {2021})}\BibitemShut {NoStop}%
    \bibitem [{\citenamefont {Bacon}\ \emph {et~al.}(2015)\citenamefont {Bacon}, \citenamefont {Flammia}, \citenamefont {Harrow},\ and\ \citenamefont {Shi}}]{bacon2015sparse}%
      \BibitemOpen
      \bibfield  {author} {\bibinfo {author} {\bibfnamefont {D.}~\bibnamefont {Bacon}}, \bibinfo {author} {\bibfnamefont {S.~T.}\ \bibnamefont {Flammia}}, \bibinfo {author} {\bibfnamefont {A.~W.}\ \bibnamefont {Harrow}},\ and\ \bibinfo {author} {\bibfnamefont {J.}~\bibnamefont {Shi}},\ }\bibfield  {title} {\bibinfo {title} {Sparse quantum codes from quantum circuits},\ }in\ \href@noop {} {\emph {\bibinfo {booktitle} {Proceedings of the forty-seventh annual ACM symposium on Theory of Computing}}}\ (\bibinfo {year} {2015})\ pp.\ \bibinfo {pages} {327--334}\BibitemShut {NoStop}%
    \bibitem [{\citenamefont {Delfosse}\ and\ \citenamefont {Paetznick}(2023)}]{delfosse2023spacetime}%
      \BibitemOpen
      \bibfield  {author} {\bibinfo {author} {\bibfnamefont {N.}~\bibnamefont {Delfosse}}\ and\ \bibinfo {author} {\bibfnamefont {A.}~\bibnamefont {Paetznick}},\ }\bibfield  {title} {\bibinfo {title} {Spacetime codes of clifford circuits},\ }\href@noop {} {\bibfield  {journal} {\bibinfo  {journal} {arXiv preprint arXiv:2304.05943}\ } (\bibinfo {year} {2023})}\BibitemShut {NoStop}%
    \bibitem [{\citenamefont {Zhou}\ \emph {et~al.}(2000)\citenamefont {Zhou}, \citenamefont {Leung},\ and\ \citenamefont {Chuang}}]{zhou2000methodology}%
      \BibitemOpen
      \bibfield  {author} {\bibinfo {author} {\bibfnamefont {X.}~\bibnamefont {Zhou}}, \bibinfo {author} {\bibfnamefont {D.~W.}\ \bibnamefont {Leung}},\ and\ \bibinfo {author} {\bibfnamefont {I.~L.}\ \bibnamefont {Chuang}},\ }\bibfield  {title} {\bibinfo {title} {Methodology for quantum logic gate construction},\ }\href@noop {} {\bibfield  {journal} {\bibinfo  {journal} {Physical Review A}\ }\textbf {\bibinfo {volume} {62}},\ \bibinfo {pages} {052316} (\bibinfo {year} {2000})}\BibitemShut {NoStop}%
    \bibitem [{\citenamefont {Bennett}\ \emph {et~al.}(1993)\citenamefont {Bennett}, \citenamefont {Brassard}, \citenamefont {Cr{\'e}peau}, \citenamefont {Jozsa}, \citenamefont {Peres},\ and\ \citenamefont {Wootters}}]{bennett1993teleporting}%
      \BibitemOpen
      \bibfield  {author} {\bibinfo {author} {\bibfnamefont {C.~H.}\ \bibnamefont {Bennett}}, \bibinfo {author} {\bibfnamefont {G.}~\bibnamefont {Brassard}}, \bibinfo {author} {\bibfnamefont {C.}~\bibnamefont {Cr{\'e}peau}}, \bibinfo {author} {\bibfnamefont {R.}~\bibnamefont {Jozsa}}, \bibinfo {author} {\bibfnamefont {A.}~\bibnamefont {Peres}},\ and\ \bibinfo {author} {\bibfnamefont {W.~K.}\ \bibnamefont {Wootters}},\ }\bibfield  {title} {\bibinfo {title} {Teleporting an unknown quantum state via dual classical and einstein-podolsky-rosen channels},\ }\href@noop {} {\bibfield  {journal} {\bibinfo  {journal} {Physical review letters}\ }\textbf {\bibinfo {volume} {70}},\ \bibinfo {pages} {1895} (\bibinfo {year} {1993})}\BibitemShut {NoStop}%
    \bibitem [{Note1()}]{Note1}%
      \BibitemOpen
      \bibinfo {note} {For a review of graph states see~\cite {hein2006entanglement}.}\BibitemShut {Stop}%
    \bibitem [{\citenamefont {Hein}\ \emph {et~al.}(2006)\citenamefont {Hein}, \citenamefont {D{\"u}r}, \citenamefont {Eisert}, \citenamefont {Raussendorf}, \citenamefont {Nest},\ and\ \citenamefont {Briegel}}]{hein2006entanglement}%
      \BibitemOpen
      \bibfield  {author} {\bibinfo {author} {\bibfnamefont {M.}~\bibnamefont {Hein}}, \bibinfo {author} {\bibfnamefont {W.}~\bibnamefont {D{\"u}r}}, \bibinfo {author} {\bibfnamefont {J.}~\bibnamefont {Eisert}}, \bibinfo {author} {\bibfnamefont {R.}~\bibnamefont {Raussendorf}}, \bibinfo {author} {\bibfnamefont {M.}~\bibnamefont {Nest}},\ and\ \bibinfo {author} {\bibfnamefont {H.-J.}\ \bibnamefont {Briegel}},\ }\bibfield  {title} {\bibinfo {title} {Entanglement in graph states and its applications},\ }\href@noop {} {\bibfield  {journal} {\bibinfo  {journal} {arXiv preprint quant-ph/0602096}\ } (\bibinfo {year} {2006})}\BibitemShut {NoStop}%
    \end{thebibliography}
\end{document}